\documentclass[11pt]{article}
\usepackage{amsmath}
\usepackage{amssymb}
\usepackage{bbm}
\usepackage[margin=2cm]{geometry}
\usepackage{hyperref}
\usepackage{tikz}
\usetikzlibrary{decorations.pathmorphing,arrows}
\usepackage{braket}
\usepackage{setspace}

\newtheorem{thm}{Theorem}[section]
\newtheorem{cor}[thm]{Corollary}
\newtheorem{prop}[thm]{Proposition}

\newtheorem{lemma}[thm]{Lemma}

\newtheorem{remark}[thm]{Remark}

\newtheorem{conj}[thm]{Conjecture}

\newenvironment{proof}{{\bf Proof:}}{\hfill$\square$\vskip.5cm}

\newcommand{\R}{\mathbb{R}}

\newcommand{\C}{\mathbb{C}}

\newcommand{\Z}{\mathbb{Z}

}

\newcommand{\x}{\boldsymbol{x}}
\newcommand{\y}{\boldsymbol{y}}

\renewcommand{\Z}{{\mathbb{Z}}}

\newcommand{\Hil}{\mathcal{H}}

\begin{document} 

\date{23 November 2024}

%\SUBMITTED{May 6, 2023}
%\ACCEPTED{00}
%\KEYWORDS{Random permutations}
%\AMSSUBJ{	82B05, 82B10, 60B15}

\title{Violation of Ferromagnetic Ordering of Energy Levels in Spin Rings
for the Singlet}
\author{
%Mrigank${}$ \\
\large David Heson${}^{1}$, Shannon Starr${}^{2}$ and Jacob Thornton${}^{3}$\\
\textcolor{white}{.}\\
\small ${}^{1}$ Mississippi State University\\[-2pt]
\small Department of Physics and Astronomy\\[-2pt]
\small 355 Lee Boulevard\\[-2pt]
\small Mississippi State, MS 39762\\
\textcolor{white}{.}\\
\small ${}^{2}$ University of Alabama at Birmingham\\[-2pt]
\small Department of Mathematics\\[-2pt]
%\small University Hall, Room 4005\\[-2pt]
\small  1402 Tenth Avenue South\\[-2pt]
\small  Birmingham, AL 35294-1241\\[-2pt]
\small  \href{mailto:slstarr@uab.edu}{slstarr@uab.edu}\\
\textcolor{white}{.}\\
\small ${}^{3}$ Auburn University\\[-2pt]
\small Department of Chemical Engineering\\[-2pt]
\small Samford Hall\\[-2pt]
\small 182 S College Street\\[-2pt]
\small Auburn University, AL 36849
}

\maketitle

\abstract{We demonstrate a violation of the ``ferromagnetic ordering of energy levels''
conjecture (FOEL) for even length spin rings. The FOEL conjecture was a guess made by Nachtergaele,
Spitzer and an author for the Heisenberg model on certain graphs: a family of inequalities,
the first of which is the statement that the spectral gap of the Heisenberg model
equals the gap of the  random walk. That first guess was originally a conjecture of Aldous which was later proved
by Caputo, Liggett and Richthammer.
We claim that for  spin rings of even length $L>4$, the lowest spin $S=0$ energy is lower than the lowest
spin $S=1$ energy. This violates the $(L/2)$-th inequality in the FOEL conjecture.

Our methodology is largely numerical: we have applied exact diagonalization up to $L=20$.
We also rigorously consider the Hamiltonian of the Heisenberg spin ring for even length $L$ projected to the spin $S=0$ sector. We prove that it has a unique ground state. Then, using the single mode approximation the uniqueness explains the energy turn-around.
Important insight comes from reconsideration of previous work by Sutherland, using the Bethe ansatz.
Especially important is a work of Dhar and Shastry that goes beyond the Bethe ansatz.}

\section{Set-up}

Let $T_L$  denote the ring graph: the graph whose vertex set is $\{1,\dots,L\}$ and with edges
between $i$ and $[i+1]$  where we denote: $[i+1]=i+1$ for $i<L$, and $[L+1]=1$.
(The letter $T$ stands for 1-dimensional torus.)
Consider the spin-$1/2$ quantum Heisenberg ferromagnet on $T_L$.
The Hamiltonian  is
\begin{equation}
\label{eq:HamDef}
	H^{\mathrm{FM}}(T_L)\, =\, \sum_{\alpha=1}^{L} h_{\alpha,[\alpha+1]}^{\mathrm{FM}}\, ,\
\text{ where }\
	h_{\alpha,[\alpha+1]}^{\mathrm{FM}}\, =\, -S^{(1)}_\alpha S^{(1)}_{[\alpha+1]} 
	- S^{(2)}_\alpha S^{(2)}_{[\alpha+1]} - S^{(3)}_\alpha S^{(3)}_{[\alpha+1]} + \frac{1}{4}\, \mathbbm{1}\, .
\end{equation}
We have shifted the energies so that the absolute minimum energy is $0$.
%(With this choice of overall multiplier each $h_{\alpha,[\alpha+1]}$ is a projection.)

We
let $\mathbbm{1}$ denote the identity on the Hilbert space.
The spin operators are the usual variations of the Pauli spin-$1/2$ matrices.
Namely on $\C^2$ we have the orthonormal basis $\Psi_{1/2,1/2}$ and $\Psi_{1/2,-1/2}$,
and in this basis we have the matrices for the spin operators
\begin{equation*}
	S^{(1)}\, =\, \begin{bmatrix} 0 & 1/2 \\ 1/2 & 0 \end{bmatrix}\, ,\qquad
	S^{(2)}\, =\, \begin{bmatrix} 0 & -i/2 \\ i/2 & 0 \end{bmatrix}\ \text{ and }\ 
	S^{(3)}\, =\, \begin{bmatrix} 1/2 & 0 \\ 0 & -1/2 \end{bmatrix}\, .
\end{equation*}
The spin-raising and lower operators of a single site are also
\begin{equation*}
	S^{+}\, =\, \begin{bmatrix} 0 & 1 \\ 0 & 0 \end{bmatrix}\ \text{ and }\ 
	S^{-}\, =\, \begin{bmatrix} 0 & 0 \\ 1 & 0 \end{bmatrix}\, .
\end{equation*}
If we denote for each $\alpha \in \{1,\dots,L\}$ the single site Hilbert space to be $\Hil_{\alpha}
\cong \C^{2}$, then the total Hilbert space for all $L$ spin sites is the tensor product
\begin{equation*}
	\Hil_{\mathrm{tot}}\, =\, \Hil_1 \otimes \Hil_2 \otimes \cdots \otimes \Hil_L\, .
\end{equation*}
Then we also denote the localization of the spin matrices
\begin{equation*}
\begin{split}
	S_{\alpha}^{(\nu)}\, &=\, \mathbbm{1}_{\C^{2j+1}} \otimes \cdots \otimes \mathbbm{1}_{\C^{2j+1}} 
\otimes S^{(\nu)} \otimes \mathbbm{1}_{\C^{2j+1}}  \otimes \cdots \otimes \mathbbm{1}_{\C^{2j+1}}\, ,\\[-15pt]
&\hspace{0.75cm} \begin{tikzpicture} \draw (0,0) node[rotate=90] {\Huge $\Bigg\{$}; \end{tikzpicture}
\hspace{1.75cm} \begin{tikzpicture} \draw (0,0) node[rotate=90] {\Huge $\Bigg\{$}; \end{tikzpicture}\\[-30pt]
&\hspace{1.65cm} \alpha-1
\hspace{3.8cm} L-\alpha
\end{split}
\end{equation*}
for $\nu \in \{1,2,3\}$,
where all factors are identities on $\C^{2}$, namely $\mathbbm{1}_{\C^{2}}$, except for the factor at position $\alpha$ which is $S^{(\nu)}$.
We also define $S^{\pm}_{\alpha} = S^{(1)}_{\alpha} \pm i S^{(2)}_\alpha$.
Another form of the Hamiltonian, equivalent to (\ref{eq:HamDef})
is the formula 
\begin{equation*}
	H^{\mathrm{FM}}(T_L)\, =\, \sum_{\alpha=1}^{L} \bigg(\frac{1}{4}\, \mathbbm{1}
-S^{(3)}_\alpha S^{(3)}_{[\alpha+1]}
-\frac{1}{2}\, S^{+}_\alpha S^{-}_{[\alpha+1]} 
-\frac{1}{2}\, S^{-}_\alpha S^{+}_{[\alpha+1]} \bigg)\, .
\end{equation*}
On the tensor product Hilbert space $\Hil_{\mathrm{tot}}^{(j)}$, 
we write the total spin operators and the Casimir operator (total spin operator):
\begin{equation*}
	S^{(\nu)}_{\mathrm{tot}}\, =\, \sum_{\alpha=1}^{L} S_\alpha^{(\nu)}\, ,\  \text{ for $\nu \in \{1,2,3\}$, and }\ 
	\mathcal{C}_{\mathrm{tot}}\, =\, 
\left(S^{(1)}_{\mathrm{tot}}\right)^2
+ \left(S^{(2)}_{\mathrm{tot}}\right)^2
+ \left(S^{(3)}_{\mathrm{tot}}\right)^2\, .
\end{equation*}
Each of the total spin operators commutes with each nearest-neighbor interaction
term $h_{\alpha,[\alpha+1]}^{\mathrm{FM}}$.
So the whole Hamiltonian commutes with $S^{(\nu)}_{\mathrm{tot}}$,
for $\nu \in \{1,2,3\}$.
Similarly, $H^{\mathrm{FM}}(T_L)$ commutes with the Casimir operator
 $\mathcal{C}_{\mathrm{tot}}$.

Henceforth, let us assume that $L$ is even.
Then we may define, for each choice of
\begin{equation*}
S \in \{0,1,2,\dots,L/2\}\, ,
\end{equation*}
the Hilbert subspace $\Hil_{\mathrm{tot}}(S)$ defined as 
\begin{equation*}
	\Hil_{\mathrm{tot}}(S)\, =\, \big\{\psi \in \Hil_{[1,L]}\, :\, \mathcal{C}_{\mathrm{tot}} \psi = S(S+1) \psi\big\}\, .
\end{equation*}
Then we define the energy eigenvalue $E_{\min}^{\mathrm{FM}}(S;T_L)$ to be
\begin{equation}
\label{eq:EminDef}
	E_{\min}^{\mathrm{FM}}(S;T_L)\, =\, \min\left(\left\{\frac{\langle \psi\, ,\ H^{\mathrm{FM}}(T_L)\, \psi \rangle}{\|\psi\|^2}\, :\, 
	\psi \in \Hil_{\mathrm{tot}}(S)\, ,\ \|\psi\|\neq 0\right\}\right)\, .
\end{equation}

\section{Main results}
We applied exact diagonalization (Lanczos iteration) to numerically determine these numbers.
\begin{figure}
{\footnotesize
$$
\begin{tikzpicture}[xscale=0.75,yscale=0.85]
\draw (0,0.75) node[below] {\begin{minipage}{1cm} Spin\\ $S$\end{minipage}}; 
%\draw (0.5,0.75) node[below] {\scriptsize \begin{minipage}{1cm} L\end{minipage}}; 
\draw (2,0) node[above] {$L=4$}; 
\draw (2,-0.75) node[] {1.00000}; 
\draw (2,-1.5) node[] {1.00000}; 
\draw (2,-2.25) node[] {0.00000}; 
\begin{scope}[xshift=2cm]
\draw (2,0) node[above] {$L=6$}; 
\draw (2,-0.75) node[] {0.69722}; 
\draw (2,-1.5) node[] {0.71922}; 
\draw (2,-2.25) node[] {0.50000}; 
\draw (2,-3) node[] {0.00000}; 
\begin{scope}[xshift=2cm]
\draw (2,0) node[above] {$L=8$}; 
\draw (2,-0.75) node[] {0.53949}; 
\draw (2,-1.5) node[] {0.55427}; 
\draw (2,-2.25) node[] {0.47431}; 
\draw (2,-3) node[] {0.29289}; 
\draw (2,-3.75) node[] {0.00000}; 
\begin{scope}[xshift=2cm]
\draw (2,0) node[above] {$L=10$}; 
\draw (2,-0.75) node[] {0.44108}; 
\draw (2,-1.5) node[] {0.45055}; 
\draw (2,-2.25) node[] {0.41331}; 
\draw (2,-3) node[] {0.32767}; 
\draw (2,-3.75) node[] {0.19098}; 
\draw (2,-4.5) node[] {0.00000}; 
\begin{scope}[xshift=2cm]
\draw (2,0) node[above] {$L=12$}; 
\draw (2,-0.75) node[] {0.37343}; 
\draw (2,-1.5) node[] {0.37970}; 
\draw (2,-2.25) node[] {0.35961}; 
\draw (2,-3) node[] {0.31258}; 
\draw (2,-3.75) node[] {0.23774}; 
\draw (2,-4.5) node[] {0.13397}; 
\draw (2,-5.25) node[] {0.00000}; 
\begin{scope}[xshift=2cm]
\draw (2,0) node[above] {$L=14$}; 
\draw (2,-0.75) node[] {0.32393}; 
\draw (2,-1.5) node[] {0.32827}; 
\draw (2,-2.25) node[] {0.31628}; 
\draw (2,-3) node[] {0.28776}; 
\draw (2,-3.75) node[] {0.24236}; 
\draw (2,-4.5) node[] {0.17962}; 
\draw (2,-5.25) node[] {0.09903}; 
\draw (2,-6) node[] {0.00000}; 
\begin{scope}[xshift=2cm]
\draw (2,0) node[above] {$L=16$}; 
\draw (2,-0.75) node[] {0.28610}; 
\draw (2,-1.5) node[] {0.28921}; 
\draw (2,-2.25) node[] {0.28153}; 
\draw (2,-3) node[] {0.26296}; 
\draw (2,-3.75) node[] {0.23336}; 
\draw (2,-4.5) node[] {0.19252}; 
\draw (2,-5.25) node[] {0.14020}; 
\draw (2,-6) node[] {0.07612}; 
\draw (2,-6.75) node[] {0.00000}; 
\begin{scope}[xshift=2cm]
\draw (2,0) node[above] {$L=18$}; 
\draw (2,-0.75) node[] {0.25623}; 
\draw (2,-1.5) node[] {0.25852}; 
\draw (2,-2.25) node[] {0.25332}; 
\draw (2,-3) node[] {0.24058}; 
\draw (2,-3.75) node[] {0.22022}; 
\draw (2,-4.5) node[] {0.19214}; 
\draw (2,-5.25) node[] {0.15622}; 
\draw (2,-6) node[] {0.11233}; 
\draw (2,-6.75) node[] {0.06031}; 
\draw (2,-7.5) node[] {0.00000}; 
\begin{scope}[xshift=2cm]
\draw (2,0) node[above] {$L=20$}; 
\draw (2,-0.75) node[] {0.23203}; 
\draw (2,-1.5) node[] {0.23377}; 
\draw (2,-2.25) node[] {0.23010}; 
\draw (2,-3) node[] {0.22098}; 
\draw (2,-3.75) node[] {0.20638}; 
\draw (2,-4.5) node[] {0.18625}; 
\draw (2,-5.25) node[] {0.16052}; 
\draw (2,-6) node[] {0.12911}; 
\draw (2,-6.75) node[] {0.09195}; 
\draw (2,-7.5) node[] {0.04894}; 
\draw (2,-8.25) node[] {0.00000}; 
\end{scope}
\end{scope}
\end{scope}
\end{scope}
\end{scope}
\end{scope}
\end{scope}
\end{scope}
\draw (1,0) -- (0,1);
\draw (1,0) -- (19,0);
\draw (1,0) -- (1,-8.5);
\draw (9.5,1) node[] {\small $E_{\min}^{\textrm{FM}}(S;T_L)$};
\draw (0.75,-0.75) node[left] {$0$};
\draw (0.75,-1.5) node[left] {$1$};
\draw (0.75,-2.25) node[left] {$2$};
\draw (0.75,-3) node[left] {$3$};
\draw (0.75,-3.75) node[left] {$4$};
\draw (0.75,-4.5) node[left] {$5$};
\draw (0.75,-5.25) node[left] {$6$};
\draw (0.75,-6) node[left] {$7$};
\draw (0.75,-6.75) node[left] {$8$};
\draw (0.75,-7.5) node[left] {$9$};
\draw (0.75,-8.25) node[left] {$10$};
\end{tikzpicture}
$$
\caption{We have printed  the values of the lowest energy of $H^{\mathrm{FM}}(T_L)$ for even values of $L$. 
\label{fig:GSEtable}
}
}
\end{figure}
In Figure \ref{fig:GSEtable} we report the numerical values we find for $E^{\mathrm{FM}}_{\min}(S;T_L)$ 
for all choices of $S$ and even $L$ up to $20$.
By the Rayleigh-Ritz variational formula, that is the smallest eigenvalue of $H^{\mathrm{FM}}(T_L)$
when restricted to the invariant subspace $\Hil_{\mathrm{tot}}(S)$.
%When there is no confusion, we may write $E_{\min}^{\mathrm{FM}}(S)$ in place of  $E_{\min}^{\mathrm{FM}}(S;L)$.

\begin{conj} 
For even length $L\geq 4$, we have $E^{\mathrm{FM}}_{\min}(0;T_L)\leq E^{\mathrm{FM}}_{\min}(1;T_L)$ with strict inequality if $L>4$.
\end{conj}

The evidence for this is that numerically it is verified for $L=4,6,8,\dots,20$, as is seen by consulting
the first two rows of the table in Figure \ref{fig:GSEtable}.

Let us 
%digress to 
state a result known as ``ferromagnetic ordering of energy levels,''
for spin chains with open boundary conditions by Nachtergaele, Spitzer and one of the authors.
Let the open chain graph be $C_L$: vertex set  $\{1,\dots,L\}$ and edge
set those pairs $(\alpha,\alpha+1)$ for $1\leq \alpha<L$. But $L$ does not share an edge with $1$.
(So no periodic boundary condition.)
Then the Hamiltonian on this graph is 
$$
	H^{\mathrm{FM}}(C_L)\, =\, \sum_{\alpha=1}^{L-1} h_{\alpha,\alpha+1}^{\mathrm{FM}}\, .
$$

\begin{prop}[Nachtergaele, Spitzer, S \cite{NSSfoel}]
Defining the energies of the open chain Hamiltonian
$$
	E_{\min}^{\mathrm{FM}}(S;C_L)\, =\, \min\left(\left\{\frac{\langle \psi\, ,\ H^{\mathrm{FM}}(C_L)\, \psi \rangle}{\|\psi\|^2}\, :\, 
	\psi \in \Hil_{\mathrm{tot}}(S)\, ,\ \|\psi\|\neq 0\right\}\right)\, ,
$$
one has $E_{\min}^{\mathrm{FM}}(S;C_L) < E_{\min}^{\mathrm{FM}}(S+1;C_L)$
for each $S=0,1,\dots,(L/2)-1$.
\end{prop}
This result was called ``ferromagnetic ordering of energy levels.'' 
This is analogous to the celebrated ``ordering of energy levels'' result of Lieb and Mattis
which is mainly applicable to antiferromagnets \cite{LiebMattisOEL}.
Based on the result, the authors stated a family of inequalities which may be satisfied
by the Heisenberg model on certain graphs $\mathcal{G} = (\mathcal{V},\mathcal{E})$.
For 
$$
H^{\mathrm{FM}}(\mathcal{G})\, =\, \sum_{\{\alpha,\beta\} \in \mathcal{E}} h^{\mathrm{FM}}_{\alpha,\beta}\, ,
$$
and 
$$
	E_{\min}^{\mathrm{FM}}(S;\mathcal{G})\, =\, \min\left(\left\{\frac{\langle \psi\, ,\ H^{\mathrm{FM}}(\mathcal{G})\, \psi \rangle}{\|\psi\|^2}\, :\, 
	\psi \in \Hil_{\mathrm{tot}}(S)\, ,\ \|\psi\|\neq 0\right\}\right)\, ,
$$
the FOEL inequality at level $n$ is the inequality
$$
	E_{\min}^{\mathrm{FM}}\left(\frac{1}{2}|\mathcal{V}|-n;\mathcal{G}\right)\, <\, 
	E_{\min}^{\mathrm{FM}}\left(\frac{1}{2}|\mathcal{V}|-n+1;\mathcal{G}\right)\, .
$$
So the FOEL-$n$ inequality is satisfied for all $n=1,\dots,(L/2)$ for the case of the graph $\mathcal{G}$
being an open chain $C_L$.
But our conjecture is that it is violated for the case of the spin ring $T_L$ for even $L\geq 4$
for the case of the FOEL-$n$ inequality with $n=L/2$.
This violation was previously noted by Spitzer, Tran and an author in \cite{SpitzerStarrTran}.
But there the evidence was more limited. Also, here we believe that there is further evidence
that the violation of FOEL is limited to the spin singlet.

We will state a bit more about the history of the FOEL conjecture at the end of this section, since 
there is an important reference that should be described.
But first, let us state one rigorous result and theoretical explanation for the reason that FOEL is violated
in the spin ring at the spin singlet.

\subsection{Non-degeneracy of the spin singlet ground state}

The Hulth\'en bracket basis, as described by Temperley and Lieb in \cite{TemperleyLieb} is as follows. 
Suppose that $\Big((\alpha_1,\beta_1),\dots,(\alpha_{L/2},\beta_{L/2})\Big)$
are a family of left and right endpoints of a perfect matching of $\{1,\dots,L\}$ simply meaning 
$\{\alpha_1,\dots,\alpha_{L/2},\beta_1,\dots,\beta_{L/2}\} = \{1,\dots,L\}$.
For definiteness let us assume $\alpha_1<\dots<\alpha_{L/2}$.
Then we also assume
\begin{itemize}
\item $\alpha_k<\beta_k$ for each $k$,
\end{itemize}
and we assume the non-crossing condition
\begin{itemize}
\item for $j$ and $k$ distinct: $\alpha_k<\alpha_j<\beta_k$ if and only if $\alpha_k<\beta_j<\beta_k$.
\end{itemize}
Let us abbreviate this as $(\boldsymbol{\alpha},\boldsymbol{\beta})$.
Let $\mathcal{M}_L$ be the set of all such non-crossing perfect matchings of $\{1,\dots,L\}$.
We define
\begin{equation*}
\begin{split}
\Phi_L(\boldsymbol{\alpha},\boldsymbol{\beta})\, 
&=\, \left(\frac{1}{\sqrt{2}}\, S_{\alpha_1}^- - \frac{1}{\sqrt{2}}\, S_{\beta_1}^{-}\right)
\cdots \left(\frac{1}{\sqrt{2}}\, S_{\alpha_{L/2}}^- - \frac{1}{\sqrt{2}}\, S_{\beta_{L/2}}^{-}\right)
\, \cdot \Big(\Psi_{1/2,1/2} \otimes \cdots \otimes \Psi_{1/2,1/2}\Big)\, .\\[-10pt]
&\hspace{9.25cm} \begin{tikzpicture} \draw (0,0) node[rotate=90] {\Huge $\Bigg\{$}; \end{tikzpicture}\\[-25pt]
&\hspace{10.55cm} L
\end{split}
\end{equation*}
In other words this is simply a product of $L/2$ spin singlets, such that the $k$-th
spin singlet is shared between spin sites $\alpha_k$ and $\beta_k$.

\begin{lemma}[Perron-Frobenius result]
\label{lem:first}
For even $L$ and $S=0$, there is a vector $\Theta_{L}$ which is positive in the Temperley-Lieb (Hulthen bracket) basis such that 
$$
	H_{L}^{\mathrm{FM}} \Theta_L\, =\, E_{\min}^{\mathrm{FM}}(0;L) \Theta_L\, .
$$
Moreover, this eigenvector is non-degenerate:
$$
	\left\{\psi \in \Hil_{[1,L]}^{(0)}\, :\,  H_{L}^{\mathrm{FM}} \psi = E_{\min}^{\mathrm{FM}}(0;L) \psi\right\}\,
	=\, \operatorname{span}(\Theta_L)\, .
$$
\end{lemma}

We will not give the full proof of the lemma here.
The proof is similar to the proof of the original ``ordering of energy levels'' result
of Lieb and Mattis \cite{LiebMattisOEL}, using the Perron-Frobenius theorem.
It is also similar to the result of Nachtergaele, Spitzer and an author, which used the Hulth\'en bracket basis.
But we do give a summary of the proof idea in Appendix \ref{app:TL}.
The full proof  is available on the arXiv in an earlier preprint version of this article.
All details are there.

The point of Lemma \ref{lem:first} is to demonstrate a symmetry property, namely non-degeneracy with respect to $H^{\mathrm{FM}}_{L}$.
That is not present in the other total spin sectors for two reasons.
Firstly, a higher spin multiplet is by definition  multiply-degenerate.
But more importantly, the lowest energy eigenvectors
may be chosen to be simultaneous eigenvectors of the translation operator,
$$
\mathcal{T} : \Hil_{[1,L]} \to \Hil_{[1,L]}
$$
defined by linearly extending the reflection operator
$$
	\mathcal{T} (\psi_1 \otimes \psi_2 \otimes \cdots \otimes \psi_{L-1} \otimes
\psi_L)\, =\, \psi_{2} \otimes \psi_{3} \otimes \cdots \otimes \psi_L \otimes \psi_1\, .
$$
In \cite{DharShastry},
Dhar and Shastry claim that such an eigenvector, with total spin $S$, may be chosen to satisfy $\mathcal{T} \psi = e^{i\phi} \psi$
with $\phi$ either equal to $\phi_+(S)$ or $\phi_{-}(S)$ where $\phi_{\pm}(S) = \pm \pi\cdot (1-(2S/L))$.

In the next section, we will give a more complete summary of their results.
For now, note that, unless $S=0$ so that $e^{i\phi_{\pm}(S)}=e^{\pm i\pi}=-1$, or $S=1$ so that $e^{i\phi_{\pm}(S)}=e^{0}=1$,
the eigenvectors Dhar and Shastry describe may be written as a pair $\psi_+(S)$, $\psi_-(S)$ such that
$\mathcal{T} \psi_{\pm}(S) = e^{i\phi_{\pm}(S)} \psi$.
Also we have  $\mathcal{R} \psi_{\pm} = \psi_{\mp}$ for the reflection operator
$$
\mathcal{R} : \Hil_{[1,L]} \to \Hil_{[1,L]}
$$
defined by linearly extending the reflection operator
$$
	\mathcal{R} (\psi_1 \otimes \psi_2 \otimes \cdots \otimes \psi_{L-1} \otimes
\psi_L)\, =\, \psi_{L} \otimes \psi_{L-1} \otimes \cdots \otimes \psi_2 \otimes \psi_1\, .
$$
We note that part of Dhar and Shastry's analysis involves stating that at $S=0$,
there is a unique ground state (of the Hamiltonian restricted to the spin $S$ subspace),
which is an eigenvector of $\mathcal{T}$ with eigenvalue $-1$ and which is an eigenvector of $\mathcal{R}$.
(The eigenvalue of $\mathcal{R}$ is $(-1)^{L/2}$.) These properties may be seen to be verified by Lemma \ref{lem:first},
using properties of the Hulth\'en bracket basis.
For now, let us note an implication.

\subsection{The single mode approximation}

Define the momentum-$\ell$ spin raising- and lowering-operators
\begin{equation*}
	\widehat{S}^{\pm}_\ell\, =\, \sum_{\alpha=1}^{L} e^{2 \pi i \ell \alpha/L} S_{\alpha}^{\pm}\, ,
\end{equation*}
for $\ell \in \Z$.
Then one version of the single mode approximation consists of starting with an energy eigenvector such
as $\Theta_L$ and considering trial wave vectors of the form
$$
	\widehat{S}^{\pm}_{\ell} \Theta_L\, ,
$$
for various choices of $\ell \in \{1,\dots,L-1\}$. (Note that $\ell=0$ gives the usual spin raising operator, which
is a symmetry of the Hamiltonian, thus excluded from consideration.)
See, for example, Chapter 9 of Auerbach \cite{Auerbach}.

\begin{cor}
\label{cor:SMA}[The SMA at the singlet]
Assuming that $L$ is even, define the SMA trial energy at momentum $\ell \in \{1,\dots,L-1\}$ to be
$$ 
\mathcal{E}^{\mathrm{SMA}}_{L}(\ell)\,
=\, 
\frac{\left\langle \Theta_L\, ,\ 
\widehat{S}^-_{-\ell} \left(H^{\mathrm{FM}}_{L}
-E^{\mathrm{FM}}_{\min}(0;L)\right) \widehat{S}^+_{\ell} \Theta_L\right\rangle}
{\left\|\widehat{S}^+_{\ell} \Theta_L\right\|^2}\, .
$$
Then
$$
\mathcal{E}_{L}^{\mathrm{SMA}}(\ell)\, =\, 
\frac{8}{3}\, 
\left(\frac{L}{4} - E^{\mathrm{FM}}_{\min}(0;L)\right)\,
\frac{\sin^2\left(\pi \ell/L\right)}{\left\|\widehat{S}^+_{\ell} \Psi^{\mathrm{FM}}_{\min}(0;L)\right\|^2}\, .
$$
%\begin{equation*}
%E^{\mathrm{FM}}_{\min}(1;L)-E^{\mathrm{FM}}_{\min}(0;L)\, \leq\,
%\mathcal{E}^{\mathrm{SMA}}_{\min}(0;L)\, ,
%\end{equation*}
%where $\mathcal{E}^{\mathrm{SMA}}_{\min}(0;L)$ is the single mode approximation
%$$ 
%\mathcal{E}^{\mathrm{SMA}}_{\min}(0;L)\,
%=\, \min_{\ell=1,\dots,L-1}
%\frac{\left\langle \Psi^{\mathrm{FM}}_{\min}(0;L)\, ,\ 
%\widehat{S}^-_{-\ell} \left(H^{\mathrm{FM}}_{L}
%-E^{\mathrm{FM}}_{\min}(0;L)\right) \widehat{S}^+_{\ell} \Psi^{\mathrm{FM}}_{\min}(0;L)\right\rangle}
%{\left\|\widehat{S}^+_{\ell} \Psi^{\mathrm{FM}}_{\min}(0;L)\right\|^2}\, .
%$$
%In particular $\mathcal{E}^{\mathrm{SMA}}_{\min}(0;L)$ is positive
%whenever $E^{\mathrm{FM}}_{\min}(0;L)<L/4$.
\end{cor}
The quantity in the denominator will be denoted by us as 
$$
	Z_L(\ell)\, =\, \left\|\widehat{S}^+_{\ell} \Psi^{\mathrm{FM}}_{\min}(0;L)\right\|^2\, ,
$$
so that we may write
$$
\mathcal{E}_L^{\mathrm{SMA}}(\ell)\, =\, 
\frac{8}{3}\, 
\left(\frac{L}{4} - E^{\mathrm{FM}}_{\min}(0;L)\right)\,
\frac{\sin^2\left(\pi \ell/L\right)}{Z_L(\ell)}\, .
$$
From its definition, $Z_L(\ell)>0$.
Numerically, it appears to be maximized at $\ell=1$ and $L-1$.
See   Figure \ref{fig:StructureFactor}.
\begin{figure}
$$
\begin{tikzpicture}[xscale=0.75,yscale=0.06]
	\draw[very thick,->] (0,-1) -- (0,50) node[above] {\small $Z_L(\ell)$};
	\draw[very thick,->] (-10.5,0) -- (11,0) node[right] {\small $\ell$};
	\draw 
(   -9.0000,    6.6663) node[] {\small $\bullet$} --
(   -8.0000 ,   6.6659)  node[] {\small $\bullet$} --
(   -7.0000  ,  6.6657)  node[] {\small $\bullet$} --
(   -6.0000   , 6.6641)  node[] {\small $\bullet$} --
(   -5.0000    ,6.6636)  node[] {\small $\bullet$} --
(   -4.0000,    6.6442)  node[] {\small $\bullet$} --
(   -3.0000 ,   6.7201)  node[] {\small $\bullet$} --
(   -2.0000  ,  4.9590)  node[] {\small $\bullet$} --
(   -1.0000   ,45.0180)  node[] {\small $\bullet$} --
(         0,         0)  node[] {\small $\bullet$} --
(    1.0000,   45.0180)  node[] {\small $\bullet$} --
(    2.0000 ,   4.9590)  node[] {\small $\bullet$} --
(    3.0000  ,  6.7201)  node[] {\small $\bullet$} --
(    4.0000   , 6.6442)  node[] {\small $\bullet$} --
(    5.0000    ,6.6636)  node[] {\small $\bullet$} --
(    6.0000,    6.6641)  node[] {\small $\bullet$} --
(    7.0000 ,   6.6657)  node[] {\small $\bullet$} --
(    8.0000  ,  6.6659)  node[] {\small $\bullet$} --
(    9.0000   , 6.6663)  node[] {\small $\bullet$} --
(   10.0000   , 6.6662) node[] {\small $\bullet$};
	\draw (0.25,20/3) -- (-0.25,20/3) node[left] {\small $20/3$};
	\draw (0.25,45) -- (-0.25,45) node[left] {\small $45$};
	\draw (10,2) -- (10,-2) node[below] {\small $10$};
	\draw (-10,2) -- (-10,-2) node[below] {\small $-10$};
\end{tikzpicture}
$$
\caption{
Numerical calculation of  $Z_L(\ell) = \|\widehat{S}^+_{\ell}  \Theta_L\|^2$ plotted against $\ell$ for $L=20$.
The range for $\ell$ is $-9$ to $10$, as plotted.
%We know that, setting $\ell=0$, we have $\|\widehat{S}^+_{0}  \Psi^{\mathrm{FM}}_{\min}(0;L)\|^2=0$
%by $\mathrm{SU}(2)$ symmetry of the singlet state 
%$\Psi^{\mathrm{FM}}_{\min}(0;L)$. 
%An easy calculation shows $\sum_{\ell=1}^{L} \|\widehat{S}^+_{\ell} \Psi^{\mathrm{FM}}_{\min}(0;L)\|^2$ equals 
%$L^2/2$ and 
%$\sum_{\ell=1}^{L} \cos(2\pi \ell/L)\|\widehat{S}^+_{\ell} \Psi^{\mathrm{FM}}_{\min}(0;L)\|^2$
%equals
%$(2L/3) \cdot \left(\frac{L}{4} - E^{\mathrm{FM}}_{\min}(0;L)\right)$.
\label{fig:StructureFactor}
}
\end{figure}
A reasonable guess, based on analysis by Dhar and Shastry that we summarize in the next section,
is that $Z_L(1)=Z_L(L-1)\sim L^2/8$.
One may prove that for sufficiently large $L$, one has $E^{\mathrm{FM}}_{\min}(0;L)<L/4$.
(See an earlier preprint version of
this article on the arXiv for full details of a proof that $E^{\mathrm{FM}}_{\min}(0;L)<L/4$
for sufficiently large $L$.)
Empirically, this seems to be true as long as $L>4$.
Therefore, all the trial energies according to the SMA are positive, with the smallest occurring for $\ell=1$ or $L-1$.

In fact, $E^{\mathrm{FM}}_{\min}(0;L)$ is asymptotic to $\pi^2/(2L)$, according to the same
analysis of Dhar and Shastry.
So we would end up with a guess  for the SMA minimum energy as 
$\mathcal{E}_{L}^{\mathrm{SMA}}(1)=\mathcal{E}_{L}^{\mathrm{SMA}}(L-1) \sim 8\pi^2/(3L^3)$.
Numerically this seems to be approximately 4 times too large.
See Figure \ref{fig:DataAn}.
In other words, the SMA seems to give a gap $E_{\min}^{\mathrm{FM}}(1;L)-E_{\min}^{\mathrm{FM}}(0;L)$
which is about 4 times too large.
But it has the correct sign.
\begin{figure}
\begin{center}
{\scriptsize
\begin{equation*}
\begin{tabular}{r|c|c|c|c|c|c|c|c|c}
$L$  &
4 &
6 &
8 &
10 &
12 &
14 &
16 &
18 & 
20\\
%\hline
%$Z(1;1/2,L)$
%&
%4
%&
%6.773501
%&
%10.228781
%&
%14.356261
%&
%19.153123
%&
%24.618190
%&
%30.750894
%&
%37.550908
%&
%45.018035\\
\hline
$E^{\mathrm{FM}}_{\min}(1;L)-
E^{\mathrm{FM}}_{\min}(0;L)$   &
0  &  
0.021999  & 
0.0147800 &
0.0094665  &  
0.0062755  &  
0.0043340  &  
0.0031045  &  
0.0022945  &
0.0017415\\
\hline
$\mathcal{E}_L^{\mathrm{SMA}}(1)$ &
                   0 &
   0.079011&
   0.0557607&
   0.0365200&
   0.0244969&
   0.0170351&
   0.0122578&
   0.0090874&
   0.0069116
\end{tabular}
\end{equation*}
}
\end{center}
\caption{
The difference $E^{\mathrm{FM}}_{\min}(1;L) - E^{\mathrm{FM}}_{\min}(0;L)$, calculated numerically.
The quantity $\varepsilon_{\mathrm{SMA}}(1/2,L)$ is equal to the 
SMA: $(8/3) \sin^2(\pi/L) 
(\frac{1}{4}L^2 - E^{\mathrm{FM}}_{\min}(0;L))/\|\widehat{S}^+_{1}\Psi^{\mathrm{FM}}_{\min}(0;L)\|^2$.
It is approximately $4$ times too large.
\label{fig:DataAn}}
\end{figure}

Let us quickly explain why Corollary \ref{cor:SMA} follows from Lemma \ref{lem:first}.
%
%It is easy to see that $E^{\mathrm{FM}}_{\min}(0;L)<L/4$
%for even values of $L$ greater than $4$, numerically.
%It can also be proved theoretically for sufficiently large $L$.
%(See an earlier preprint version of
%this article on the arXiv for full details of a proof that $E^{\mathrm{FM}}_{\min}(0;L)<L/4$.)
%
We know that
$$
\left(\widehat{S}^{+}_{\ell}\right)^*\, =\, \widehat{S}^-_{-\ell}\, .
$$
But also, defining spin-flip symmetry by the unitary $\mathcal{F} = \exp\left(i \pi S_{[1,L]}^{(1)}\right)$, we have
$$
\widehat{S}^{-}_{-\ell}\, =\, \mathcal{R} \mathcal{F} \widehat{S}^{+}_{\ell}
\mathcal{F}\mathcal{R}\, .
$$
So the following corollary is applicable.
\begin{cor}
\label{cor:Pert}
Suppose that $A$ is an operator on $\Hil_{[1,L]}$ such that there is some unitary operator
$U$ on $\Hil_{[1,L]}$, satisfying
\begin{itemize}
\item $A^*=U^*AU$, 
\item $H^{\mathrm{FM}}_L = U^* H^{\mathrm{FM}}_L U$, and
\item $\mathcal{C}_{[1,L]} = U^* \mathcal{C}_{[1,L]} U$.
\end{itemize}
Then we have the double-commutator formula for the first-order perturbation term:
$$
\langle \Theta_L\, ,\ A^* 
\big(H_{L}^{\mathrm{FM}}-E^{\mathrm{FM}}_{\min}(0;L)\big)
A\,  \Theta_L\rangle\, 
=\, \frac{1}{2}\, \langle \Theta_L\, ,\ [A^*,[H_L^{\mathrm{FM}},A]]\,  \Theta_L\rangle\, .
$$
\end{cor}
%\begin{remark}
%The single-mode-approximation consists of calculating the terms 
%$$
%\langle \Psi_L^{(0)}\, ,\ \widehat{S}^{-}_{-\ell} 
%\big(H_{L}^{\mathrm{FM}}-E^{\mathrm{FM}}_{\min}(0;L)\big)
%\widehat{S}^+_{\ell} \Psi_L^{(0)}\rangle\, .
%$$
%For example, see Chapter 9 from the textbook by Assa Auerbach \cite{Auerbach}.
%\end{remark}
\begin{proof}
Because of the properties $U\Theta_L$ is an eigenvector of 
$H^{\mathrm{FM}}_{L}$ with eigenvalue 
$E^{\mathrm{FM}}_{\min}(0;L)$, and $U\Theta_L$ is in $\Hil_{[1,L]}^{(0)}$.
Therefore, by Lemma \ref{lem:first}, we have 
$U\Theta_L = e^{i \theta} \Theta_L$ for some $\theta \in \R$.
Thus,
\begin{equation*}
\begin{split}
	\langle \Theta_L\, ,\ A \big(H^{\mathrm{FM}}_{L}
	-E^{\mathrm{FM}}_{\min}(0;L)\big) A^* \, \Theta_L\rangle\,
	=\,
	\langle \Theta_L\, ,\ U^* A^* \big(H^{\mathrm{FM}}_{L}
	-E^{\mathrm{FM}}_{\min}(0;L)\big) A U\, \Theta_L\rangle\\
	=\, 	\langle \Theta_L\, ,\ A^* \big(H^{\mathrm{FM}}_{L}
	-E^{\mathrm{FM}}_{\min}(0;L)\big) A\, \Theta_L\rangle\,
\end{split}
\end{equation*}
But a trivial calculation shows that 
\begin{equation*}
\begin{split}
	\left\langle \Theta_L\, ,\ \left[A^*\, ,\ \left[\big(H^{\mathrm{FM}}_{L}
	-E^{\mathrm{FM}}_{\min}(0;L)\big)\, ,\ A\right]\right] \,
	\Theta_L\right\rangle\,
	&=\,
	\langle \Theta_L\, ,\ A \big(H^{\mathrm{FM}}_{L}
	-E^{\mathrm{FM}}_{\min}(0;L)\big) A^*\, \Theta_L\rangle\\
	&\qquad+	\langle \Theta_L\, ,\ A^* \big(H^{\mathrm{FM}}_{L}
	-E^{\mathrm{FM}}_{\min}(0;L)\big) A\, \Theta_L\rangle\, ,
\end{split}
\end{equation*}
owing to the fact that $H^{\mathrm{FM}}_{L}
	-E^{\mathrm{FM}}_{\min}(0;L)$ annihilates $\Theta_L$.
Therefore, this proves the corollary because $[A^*,[H-c\mathbbm{1},A]]$
is independent of $c$.
\end{proof}
To prove Corollary \ref{cor:SMA}, apply this with the unitary $U = \mathcal{R} \mathcal{F} = \mathcal{F} \mathcal{R}$
which satisfies $U^*=U$.
%As we mentioned before, the single mode approximation (SMA)
%is a popular perturbation scheme which is used because of its
%ease of calculation.
A standard exercise shows that, for $q_L=e^{2\pi i/L}$:
\begin{equation}
\label{eq:DoubleCommFinal}
\left[\widehat{S}_{-\ell}^-,  \left[\widehat{S}_{\ell}^+,H_{L}^{\mathrm{FM}}\right]\right]\,
	=\,  -\sum_{\alpha=1}^{L} \Big(2(1-q_L^\ell)(1-q_L^{-\ell}) S_\alpha^{(3)} S_{\alpha+1}^{(3)} 
+(1-q_L^\ell) S_\alpha^- S_{\alpha+1}^+
+(1-q_L^{-\ell})S_\alpha^+ S_{\alpha+1}^{-}\Big)\, .
\end{equation}
(The interested reader may see an earlier preprint version of this article available
on the arXiv, where full details of the proof of Corollary \ref{cor:SMA} are provided.)
Since $\mathcal{T} \Theta_L = e^{i\theta} \Theta_L$
(indeed for $\theta=\pi$ by properties of the Hulth\'en bracket basis),
this gives the result, also using $\mathrm{SU}(2)$ symmetry.
%Using this, another exercise then shows that 
%\begin{equation}
%\label{eq:DoubleToDo2}
%	\left\langle \Psi^{\mathrm{FM}}_{\min}(0;L)\, ,\  \widehat{S}_{-\ell}^- \left(H_{L}^{\mathrm{FM}}-E^{\mathrm{FM}}_{\min}(0;L)\right)\widehat{S}_{\ell}^+ 
% \Psi^{\mathrm{FM}}_{\min}(0;L)\right\rangle\, 
%=\, \frac{8}{3}\, \sin^2\left(\frac{\pi \ell}{L}\right)\, 
%\left(\frac{L}{4} - E^{\mathrm{FM}}_{\min}(0;L)\right)\, ,
%\end{equation}
%assuming normalization $\|\Psi^{\mathrm{FM}}_{\min}(0;L)\|=1$.
%We will leave these exercises to the reader.
%(You may see an earlier preprint version of this article available
%on the arXiv, where full details are provided.)

\section{Review of previous results: Dhar and Shastry}

In \cite{DharShastry}, Dhar and Shastry considered the problem of finding $E^{\mathrm{FM}}_{\min}(S;L)$
for all choices of spins $S$.
One result is that they showed that the Bethe ansatz equations
for these vectors reduce to a 1-parameter generalization of the Stieltes equations familiar from solving for roots
of orthogonal polynomials.
The advantage is that their equations are explicitly algebraically solvable rational equations; whereas, the Bethe ansatz equations are transcendental integrable systems.
In addition they posited that the vectors  may be written as 
$\Xi_{L}^{\mathrm{Bloch}}(S)$ or $\mathcal{R} \Xi_L^{\mathrm{Bloch}}(S)$, where
\begin{equation*}
\begin{split}
	\Xi_{L}^{\mathrm{Bloch}}(S)\, =\, \left(\widehat{S}^{-}_{1}\right)^{(L/2)-S} \Big(\Psi_{1/2,1/2} \otimes \cdots \otimes \Psi_{1/2,1/2}\Big)\, .\\[-10pt]
&\hspace{-3.75cm} \begin{tikzpicture} \draw (0,0) node[rotate=90] {\Huge $\Bigg\{$}; \end{tikzpicture}\\[-25pt]
&\hspace{-2.45cm} L
\end{split}
\end{equation*}
They call these Bloch wall states
They argued that the limiting form of the energies are given by the formula
\begin{equation*}
	\widetilde{E}^{\mathrm{Bloch}}_{L}(S)\, =\, \frac{\pi^2}{2L}\, \cdot \left(1-\frac{4S^2}{L^2}\right)\, ,
\end{equation*}
at least asymptotically when $L \to \infty$.
One can check from  Figure \ref{fig:BlochWallApprox} that this is close to the exact formula for $L=20$.
But, as noted before, we find numerically that there is an energy turnaround at $S=0$.
In their own paper, in their Figure 2, they showed the deviation for $L=16$, where there is also
a turnaround at the spin singlet.

\begin{remark}
An interesting point is that for the spin $1/2$  XXZ spin ring with Ising like anisotropy,
\begin{equation*}
H_{[1,L]}^{(q)}\, =\, \sum_{\alpha=1}^{L} \left(\frac{1}{4}\, \mathbbm{1} - S^{(3)}_{\alpha} S^{(3)}_{\alpha+1}
-\frac{2}{q+q^{-1}}\left(S^{(1)}_{\alpha}S^{(1)}_{\alpha+1} + S^{(2)}_{\alpha} S^{(2)}_{\alpha+1}\right)\right)\, ,
\end{equation*}
with $0\leq q\leq 1$, it is known that the elementary excitations are literal droplets
of spin, with both a left and right domain wall, asymptotically when $L \to \infty$. See \cite{NachtergaeleStarrDroplet}.
Moreover, there is a closely related $\mathrm{SU}_q(2)$-symmetric version of the XXZ chain with kink or anti-kink
boundary conditions, where the droplets are still the elementary excitations, apparently created by introducing 
an interface of the opposite type: a droplet consists of both a kink and anti-kink interface at the left and right domain walls.
See \cite{NSSdroplet}.
\end{remark}

Dhar and Shastry's  Bloch wall states are not exact eigenstates of the Casimir operator $\mathcal{C}_{\mathrm{tot}}$.
But they are approximate eigenstates.
At $S=0$
there are two distinct Bloch walls,
$\Xi_{L}^{\mathrm{Bloch}}(0)$ and 
$\mathcal{R} {\Xi}_{L}^{\mathrm{Bloch}}(0)$.
But Lemma \ref{lem:first} implies that the lowest energy spin singlet
is non-degenerate.
In an earlier preprint version of this article on the arXiv we showed that the inner-product may be
expressed in terms of Gaussian binomial coefficients and is easily seen to be on the order of $1/2^L$.
We also note that similar types of calculations show that if the Bloch wall $\Xi_L^{\mathrm{Bloch}}(0)$
were exactly equal to $\Theta_L$ then we would have
$Z_L(1) = (L/2)((L/2)+1)$, while for $\ell\neq 1$ we would have  $Z_L(\ell)=(L/2)((L/2)-1)/(L-1)$
(The fact that this is not $0$ at $\ell=0$ is due to the fact that $\Xi_L^{\mathrm{Bloch}}(0)$ is not
exactly a spin $S=0$ state.)

Therefore, a standard textbook guess takes over.
One would guess that the Hamiltonian actually connects these two trial wave vectors, albeit possibly with a small matrix entry.
Then there is an energy splitting based on the symmetry.
This guess would also help to explain why $E^{\mathrm{FM}}_{\min}(0;L) < E^{\mathrm{FM}}_{\min}(1;L)$
because the slope of $\widetilde{E}^{\mathrm{Bloch}}_L(S)$ becomes $0$ at $S=0$,
so that the splitting could cause the $S=0$ eigenvalue to be a bit smaller than the $S=1$ eigenvalue.
If we accepted this guess then the new guess for $Z_L(\ell)$ would be 
$Z_L(1)=Z_L(L-1)\sim L^2/8$ while for $\ell\not\in\{1,L-1\}$ we would guess $Z_L(\ell) \sim L/4$.

This seems like a textbook description, except that the second eigenvalue from the splitting should then be just slightly
higher.
But it is not: the spectral gap in the $S=0$ sector is numerically seen to be bigger than that simple guess
would suggest.

\begin{figure}
$$
\begin{tikzpicture}[xscale=1.5,yscale=15]
\draw (10,0) node[] {\tiny $\bullet$}
-- (9,0.04894) node[] {\tiny $\bullet$}
-- (8,0.09195) node[] {\tiny $\bullet$}
-- (7,0.12911) node[] {\tiny $\bullet$}
-- (6,0.16052) node[] {\tiny $\bullet$}
-- (5,0.18625) node[] {\tiny $\bullet$}
-- (4,0.20638) node[] {\tiny $\bullet$}
-- (3,0.22098) node[] {\tiny $\bullet$}
-- (2,0.23010) node[] {\tiny $\bullet$}
-- (1,0.23377) node[] {\tiny $\bullet$}
-- (0,0.23203) node[] {\tiny $\bullet$};
\draw[line width=2pt,black!25!white,opacity=0.5] (10,0)  node[black] {\tiny $\boldsymbol{\circ}$}
-- (9,0.0469)  node[black] {\tiny $\boldsymbol{\circ}$}
-- (8,0.0888)  node[black] {\tiny $\boldsymbol{\circ}$}
-- (7,0.1258)  node[black] {\tiny $\boldsymbol{\circ}$}
-- (6,0.1579)  node[black] {\tiny $\boldsymbol{\circ}$}
-- (5,0.1851)  node[black] {\tiny $\boldsymbol{\circ}$}
-- (4,0.2073)  node[black] {\tiny $\boldsymbol{\circ}$}
-- (3,0.2245)  node[black] {\tiny $\boldsymbol{\circ}$}
-- (2,0.2369)  node[black] {\tiny $\boldsymbol{\circ}$}
-- (1,0.2443)  node[black] {\tiny $\boldsymbol{\circ}$}
-- (0,0.2467)  node[black] {\tiny $\boldsymbol{\circ}$};
\draw[very thick,->] (0,-0.05) -- (0,0.3) node[above left] {\small $E$};
\draw[very thick] (0.1,0.25) -- (-0.1,0.25) node[left] {\small $0.25$};
\draw[very thick,->] (-0.25,0) -- (10.5,0) node[below] {\small $S$};
\foreach \x in {1,2,...,10} {
	\draw[very thick] (\x,0.01) -- (\x,-0.01) node[below] {\small \x};}
\end{tikzpicture}
$$
\caption{
The black dots are the actual energies $E^{\mathrm{FM}}_{\min}(S;L)$ for $L=20$. 
The circles (connected by grey lines) are the Bloch wall approximations calculated by Shastry and Dhar to be 
$\widetilde{E}^{\mathrm{Bloch}}_{20}(S)$.
%$ = \frac{2 \pi^2}{L} \cdot \left(\frac{1}{4}-\frac{S^2}{L^2}\right)$.
The approximation is fairly close, with deviations primarily near $S=0$.
\label{fig:BlochWallApprox}}
\end{figure}
See, for instance, Figure \ref{fig:full} for the full spectrum for $L=14$, calculated numerically.
\begin{figure}[]
$$
\begin{tikzpicture}[xscale=0.9,yscale=3]
\foreach \y in {0.5125,1.4066,1.7488        }
	\draw (0,\y) node[] {$\bullet$};
\foreach \y in {0.9047,1.2002,1.3976    ,1.4579    ,1.4962, 1.6305, 1.7238     ,1.8112    ,    1.9743   }
	\draw[gray] (-0.1,\y) -- (0.1,\y) node[right] {\tiny $2$};
\foreach \y in {0.5370    ,1.2900    , 1.4372  ,1.8076       }
	\draw (1,\y) node[] {$\bullet$};
\foreach \y in {0.5170    ,0.8549    ,0.9319    ,0.9684    ,1.0966  ,1.2309    ,1.2416    ,    1.3231    ,1.3850    ,1.4166    ,1.4339   ,1.4856 ,1.5168    , 1.5435,1.5443    ,  1.5700    , 1.5794    ,1.6656    ,1.7415,1.7605    ,1.7694    ,1.7719  ,1.8439,  1.8464    ,1.8803    ,1.8974    ,1.9162    ,1.9950}
	\draw[gray] ( (0.9,\y) -- (1.1,\y) node[right] {\tiny $2$};
\foreach \y in {0.5884    ,1.5048         }
	\draw (2,\y) node[] {$\bullet$};
\foreach \y in { 0.5066    ,0.5676    ,0.7922    ,0.9120    ,0.9825,0.9901    ,1.0191    ,1.0248    ,1.0777    ,1.1629,1.2977    ,1.2992 ,1.3226  ,1.3957    ,1.4166   ,1.4343    ,1.4421    ,1.4816     ,1.5217    ,1.5443    , 1.5863    ,1.6019    ,1.6143    ,1.6572    ,1.6623    ,1.6743    ,1.7207,1.7395    ,1.7478    ,1.7800    ,1.8407    ,1.8475    ,1.8683    ,1.8713    ,1.8909    ,1.9309    ,1.9331    ,1.9407    ,1.9445    ,1.9478    ,1.9581    ,1.9708    ,1.9917}
	\draw[gray]  (1.9,\y) -- (2.1,\y) node[right] {\tiny $2$};
\foreach \y in { 1.3867        }
	\draw[gray]  (1.8,\y) -- (2.2,\y) node[right] {\tiny $4$};
\foreach \y in { 0.6725     ,0.9040 ,1.5296 ,1.6292      }
	\draw (3,\y) node[] {$\bullet$};
\foreach \y in {0.4812    ,0.5850    ,0.6502    ,0.7162    ,0.8571,0.8818    ,1.0084    ,1.0566    ,1.0880    ,1.0892    ,1.1068    ,1.1208    ,1.1961    ,1.2070    ,1.2421    ,1.2814    ,1.3864    ,1.4055    ,1.4252    ,1.4804    ,1.4941    ,1.5093,1.5191    ,   1.5317, 1.5365    ,1.5374    ,1.5996    ,1.6234    ,1.6426    ,1.6628    ,1.6697    ,1.6988    ,1.7065    ,1.7071    ,1.7076    ,1.8189    ,1.8289    ,1.8295,1.8458    ,1.8622    ,1.8673    ,1.8998  ,  1.9267    ,1.9344    ,1.9570    ,1.9683    ,1.9864    ,1.9963}
	\draw[gray]  (2.9,\y) -- (3.1,\y) node[right] {\tiny $2$};
\foreach \y in {1.8925    }
	\draw[gray]  (2.8,\y) -- (3.2,\y) node[right] {\tiny $4$};
\foreach \y in {0.8012  ,1.8701   }
	\draw (4,\y) node[] {$\bullet$};
\foreach \y in {0.4404    ,0.5890    ,0.6265    ,0.7033    ,0.7194    ,0.7761    ,0.8405    ,1.0043    ,1.0234  ,1.0834    ,1.1646    ,1.1772    ,1.2494,1.2523    ,1.2670    ,1.2795    ,1.3190    ,1.3821,1.3907    ,1.4486    ,1.4947  ,1.5257    ,  1.5472    , 1.5925    , 1.6103    ,1.6642    ,1.6734    ,1.6758    ,1.6897    ,1.6954    ,1.7056    ,1.7084    ,1.7085    ,1.7548    ,1.7571    ,1.7758    ,1.8003    ,1.8037    ,1.8218    ,1.8269    ,1.8325    ,1.8508    ,1.8526,1.8851    ,1.9173    ,1.9544    }
	\draw[gray] ( (3.9,\y) -- (4.1,\y) node[right] {\tiny $2$};
\foreach \y in { 1.0803        }
	\draw[gray]  (3.8,\y) -- (4.2,\y) node[right] {\tiny $4$};
\foreach \y in {0.5686    ,1.0001,1.2756    }
	\draw (5,\y) node[] {$\bullet$};
\foreach \y in {0.3843    ,0.5225   ,0.5795    ,0.7485    ,0.7876,0.8825    ,0.9103    ,0.9500    ,0.9698    ,1.0357    ,1.0895    ,1.2341    ,1.2398    ,1.2566    ,1.2583    ,  1.3441    ,1.3805    ,1.4344    ,1.5047,1.5084 ,1.5112    ,  1.5409    , 1.5456    ,1.5591,1.5659    ,1.6259    ,1.6439    ,1.6578    ,1.7291    ,1.7548    ,1.8210    ,1.8566    ,1.8916    ,1.8951,1.9007 ,1.9488    ,1.9507    ,  1.9809}
	\draw[gray]  (4.9,\y) -- (5.1,\y) node[right] {\tiny $2$};
\foreach \y in { 1.3333    }
	\draw (6,\y) node[] {$\bullet$};
\foreach \y in {0.3124    ,0.4039    ,0.5563    ,0.7224    ,0.7871    ,0.8048    ,1.0002    ,1.0476    ,1.0824    ,1.1482    ,1.1774    ,1.2715    ,1.2875    ,1.2931    ,  1.3799,1.4098    ,1.4974, 1.6086    , 1.7289    , 1.7370, 1.7379    ,1.7453    ,1.8781    ,1.8983    ,1.9524    ,1.9731    ,1.9944}
	\draw[gray]  (5.9,\y) -- (6.1,\y) node[right] {\tiny $2$};
\foreach \y in {0.2701   , 1.0440 ,2.0000}
	\draw (7,\y) node[] {$\bullet$};
\foreach \y in { 0.2247   ,0.5192    ,0.6442    ,0.8219    , 1.0561    ,1.1742    ,1.1982    ,1.5000   ,1.6204    , 1.7045    , 1.7660    ,1.8651    ,1.9397    }
	\draw[gray] (6.9,\y) -- (7.1,\y) node[right] {\tiny $2$};
\foreach \y in {  0.1206  ,  0.4679  ,  1.0000  ,1.6527       }
	\draw[gray] (7.9,\y) -- (8.1,\y) node[right] {\tiny $2$};
\foreach \y in {0}
	\draw (9,\y) node[] {$\bullet$};
\foreach \x in {0,1,...,9}
	\draw (\x,-0.125) node[] {\scriptsize $S=\x$};
\draw[very thick,->] (-0.5,-0.125) -- (-0.5,2.125) node[above left] {$E$};
\draw (-0.4,0) -- (-0.6,0) node[left] {\small $0$};
\draw (-0.4,0.5) -- (-0.6,0.5) node[left] {\small $0.25$};
\draw (-0.4,1) -- (-0.6,1) node[left] {\small $0.5$};
\draw (-0.4,1.5) -- (-0.6,1.5) node[left] {\small $0.75$};
\draw (-0.4,2) -- (-0.6,2) node[left] {\small $1$};
\draw[] (0,0.548311) -- (0.5,0.551704) --  (1,0.561977) -- (1.5,0.579426)-- (2,0.604575) -- (2.5,0.638236) -- (3,0.681593) --(3.5,0.736357)-- (4,0.805036) --(4.5,0.891347) -- (5,1.00101) -- (5.5,1.14316) -- (6,1.33338) -- (6.5,1.60007) -- (6.75,1.77757) -- (7,1.9995);
	\draw (14,1) node[] {\includegraphics[width=7cm]{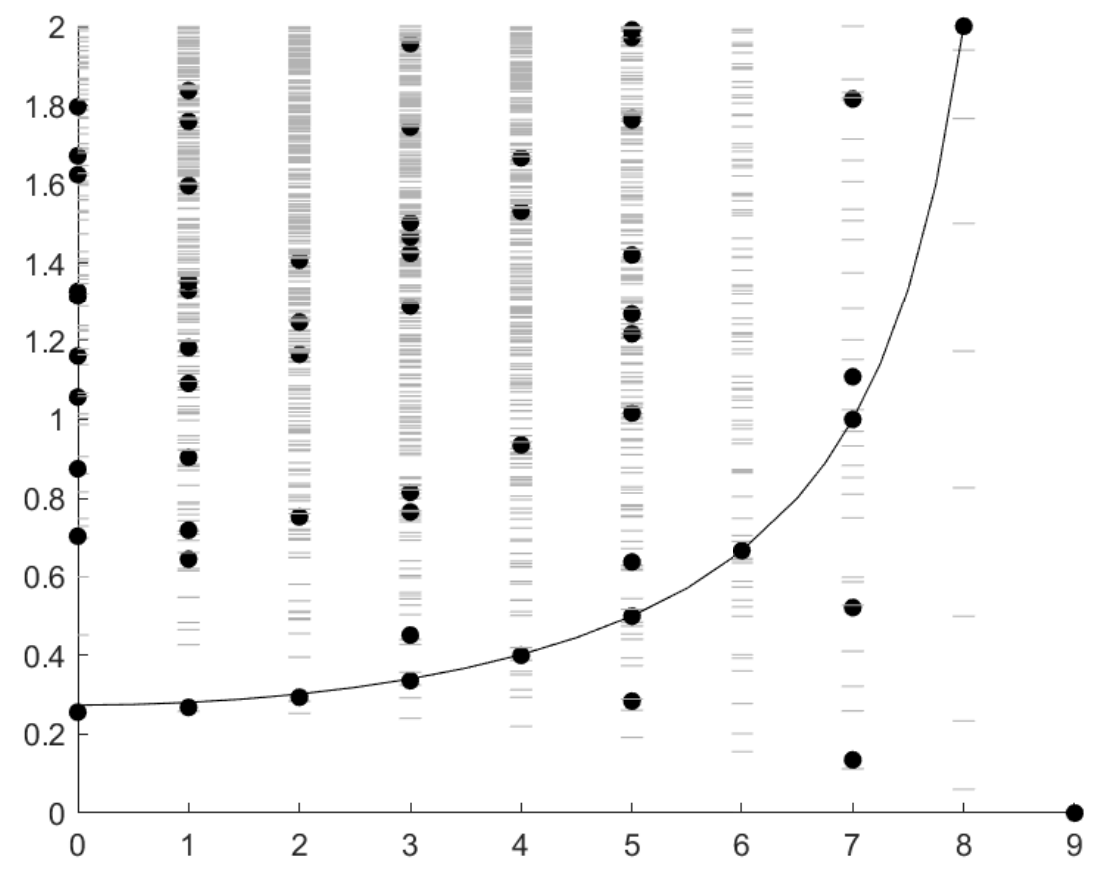}};
\end{tikzpicture}
$$
\caption{Two pictures of the full spectrum of the spin-$1/2$ spin ring for $L=18$ up to a energy cut-offs: $E=2$
for the left plot and $E=4$ for the right plot.
In the left plot, the cases of 2-fold
degeneracy are indicated with gray line segments with labels $2$. The cases of 4-fold are indicated with longer line segments with labels 4.
The non-degenerate energy levels are indicated with black bullet points.
The  curve is Sutherland's formula.
\label{fig:full}}
\end{figure}

\section{Review of previous results: Sutherland}
In \cite{Sutherland}, Bill Sutherland considered the spectrum of $H^{\mathrm{FM}}(T_L)$
using the Bethe ansatz.
In \cite{Bethe}, Hans Bethe posited a wave vector of the following form
$$
\Lambda^{\mathrm{Bethe}}_L(f_1,\dots,f_n)\, =\, \sum_{\pi \in \mathfrak{S}_n}  \sum_{1\leq \alpha_1<\dots<\alpha_n\leq L} \xi_{\pi}(\boldsymbol{f},\boldsymbol{\alpha})
\cdot \left(S_{\alpha_1}^- \cdots S_{\alpha_n}^-\right)\, \cdot
\left(\Psi_{1/2,1/2} \otimes \cdots \otimes \Psi_{1/2,1/2}\right)\, ,
$$
where, $\mathfrak{S}_n$ is the set of bijections $\pi : \{1,\dots,n\} \to \{1,\dots,n\}$ and
for $\boldsymbol{f} = (f_1,\dots,f_n)$ and $\boldsymbol{\alpha} = (\alpha_1,\dots,\alpha_n)$, we have
$$
\xi_{\pi}(\boldsymbol{f},\boldsymbol{\alpha})\, =\, \exp\left(i\sum_{1\leq j<k\leq n}\varphi_{\pi(j),\pi(k)}\right)
\exp\left(i \sum_{k=1}^{n} f_{\pi(k)} \alpha_k\right)\, .
$$
Here the constants $\varphi_{j,k}$ are chosen to solve
$$
	2\, \cot\left(\frac{\varphi_{j,k}}{2}\right)\, =\, \cot\left(\frac{f_j}{2}\right) - \cot\left(\frac{f_k}{2}\right)\, ,\ \text{ with $-\pi\leq \varphi_{j,k}<\pi$ for each $j<k$.}
$$
Then the Bethe ansatz reduces the eigenvalue problem for the Heisenberg Hamiltonian to the determination of the wavenumbers $f_1,\dots,f_n$.
Calling the numbers $\cot(f_k/2)$ the ``Bethe roots,'' these numbers solve a 
particular set of transcendental equations called the Bethe ansatz equations (BAE).
Bethe found the ``string states,'' which are characterized by the property that there is some constant $a$ such that
$$
	\cot(f_k/2)\, =\, a - i \kappa\, ,\ \text{ for $\kappa = -(n-1),-(n-3),\dots,(n-1)$.}
$$
It turns out that the total momentum is then 
$$
\theta\, =\, f_1+\dots+f_n\, ,
$$
and the constant $a$ is $n \cot(\theta/2)$. The energy of such a state is 
$$
	E\, =\, \frac{1}{n}\, (1-\cos(\theta))\, ,
$$
with our convention for the Hamiltonian.
Here $n=\frac{L}{2}-S$.

Sutherland noted that the string states break-down if $n/L$ becomes appreciable.  He pointed out that Bethe has assumed $n$ to be fixed, while thinking of $L$
large, when he found the string solutions.
Specifically, thinking of the Bethe roots in the string solutions as occupying vertical strings, Sutherland found that the Bethe roots accumulate
along other curves, which are vertical in parts -- with density of 1 Bethe root per 2 units of length in the $i$ direction.
But at certain points along the vertical string, the Bethe roots break off to follow
other curves, and with other densities.
As a particular case, he considered $\theta=\pi$, which turns out to be simplest.
Note that for $\theta=\pi$, Bethe energy relation is  $E=2/n$.
But Sutherland found this dispersion relation:
in terms of complete
elliptic integrals of the first and second kind, for $a>1$,
\begin{equation}
\label{eq:elliptic}
\begin{split}
	d\, &=\, \frac{1}{2} - \frac{a}{2}\, \left(1-\frac{E(1/a)}{K(1/a)}\right)\, ,\\
	\varepsilon\, &=\, 2 K\left(\frac{1}{a}\right)\left(2 E\left(\frac{1}{a}\right) - \left(1-\frac{1}{a^2}\right)K\left(\frac{1}{a}\right)\right)\, ,
\end{split}
\end{equation}
where $\varepsilon = E/L$
and $d=n/L \in [0,1/2]$.
Note that this is a correction to Bethe's formula in the sense that $L \varepsilon(a) \sim 2/(Ld(a))$ as $a \to 1^-$.
This can be determined by properties of the complete elliptic integrals.
Also, Sutherland noted that Taylor expanding the complete elliptic integrals,
$$
	\varepsilon\, =\, \frac{\pi^2}{2}\left(1+8\left(d-\frac{1}{2}\right)^2\right) + 
O\left(\left(d-\frac{1}{2}\right)^4\right)\ \text{ as $d \to \frac{1}{2}^-$.}
$$
He called this weak-paramagnetism, presumably because the first derivative vanishes at $d=0$, the spin singlet.
But this is weak paramagnetism just of the momentum $\pi$ sector.
Sutherland's curve agrees well with the momentum $\pi$ eigenvectors. See Figure \ref{fig:full}.
(Sutherland promised to write a longer paper fully elucidating all curves of the lowest momentum $\theta$ curves in each total spin $S$ space,
but we have not identified the follow-up paper. An important and interesting paper that did carry this idea further is more recent \cite{BBG}.)

We note that Dhar and Shastry also found corrections to Bethe's energy dispersion.
For Dhar and Shastry $\theta = 2n\pi/L$, for $n=0,1,\dots,L/2$.
But, whereas $E = \frac{1}{n}(1-\cos(2n\pi/L)) \sim 2n\pi^2/L^2$
is a very good approximation to $\widetilde{E}^{\mathrm{Bloch}}_L(\frac{L}{2}-n) = \frac{2\pi^2}{L^3} n (L-n)$ when $n\ll L$,
it deviates when $n/L$ is appreciable.
We note that Sutherland's paper on deviations from the string hypothesis was the first.

\subsection{Conclusion based on previous results}
What appears to happen in the $S=1$ sector is that the momentum $\pi$ eigenstate predicted by Sutherland is just a bit above
the $\theta=\pm(\pi - \frac{2\pi}{L})$ pair of eigenstates predicted by Dhar and Shastry's ``Bloch wave'' picture.
Thus, the elementary textbook guess for splitting of a pair of two degenerate trial wave vectors, that we stated before, 
needs to be ammended.
Instead the degeneracy is lifted by a more complicated type of perturbation theory with at least 3 vectors contributing.
Consulting Figure \ref{fig:full},
the spectral gap in the $S=0$ sector above $E^{\mathrm{FM}}_{\min}(0;L)$ appears to have a 2-fold degeneracy.
Moreover, the upward shift represents a substantially larger number than the downward deviation from Dhar and Shastry's prediction based on
the Bloch wall states.

\section{Outlook and extensions}

We have considered a certain conjectured property known as FOEL for the Heisenberg spin ring.
Note that the open chain without periodic boundary conditions was previously considered, and found to obey FOEL
without exceptions.
But the ring has 1 counterexample to FOEL: the spin singlet energy is slightly less than the spin triplet energy.
We make this claim based on exact diagonalization up to size 20, but within this numerical framework no other
counterexamples were discovered.
We have provided a rigorous theorem that, within the spin singlet sector, the restricted Hamiltonian has a 
unique ground state, positive in the Hulth\'en bracket basis.
This facilitates a double-commutator calculation of the energies in the SMA, and we find that all such energies
lead to a positive energy change.
This may be seen as a partial explanation for why the spin $S=0$ energy is less than the spin $S=1$ energy:
because the SMA changes the $S=0$ groundstate to an $S=1$ trial wave vector.

Everything in this paper was reported for spins-$1/2$.
Within this context, the Bethe ansatz results of Sutherland and later Dhar and Shastry are especially helpful.
The lowest energy eigenstate for every spin $S>0$ is given by a Bloch wall state, according to the analysis of Dhar
and Shastry.
But at $S=0$ there would be two distinct Bloch wall trial wavevectors.
So at minimum one would expect an energy splitting in perturbation theory.
Actually, there is a third vector at $S=1$ nearly intersecting the Bloch wall states: the impinging $\theta=\pi$
branch elucidated by Sutherland.
It appears that a 3-vector perturbation theory leads to a downward bend of the energy at the $S=0$ level,
plus an appreciable gap above this with a doubly degenerate excitation within the spin singlet sector.
One may note that extra hidden symmetries in the Heisenberg model may actually allow energy levels to ``cross''
due to complete integrability.

The most obvious extension is to higher spins and higher dimensions.
The original desideratum was to actually prove FOEL in dimensions $d=3$ and higher, as a step
towards trying to prove a phase transition for the Heisenberg ferromagnet.
The latter goal is an important open question. We refer to the foremost paper on the topic so far by
Correggi Giuliani and Seiringer \cite{CGS}.
The easier goal is to just understand all violations of FOEL in $d=1$ for higher spins $j$.

The numerics is more challenging for $j>1/2$.
For example, numerically the dimensions are such that if you could diagonalize a chain of length $L$ for spins-$1/2$,
then for spin $j$, the length you can diagonalize with the same set-up is $L'$ where $L' = L \cdot \ln(2)/\ln(2j+1)$.
We note that Lemma \ref{lem:first} is easy to generalize to higher spins. 
The Temperley-Lieb algebra extends to higher spins.
The higher spin version of the Hulth\'en bracket basis was already used by Nachtergaele and an author
to prove FOEL for higher spins on the open chains $C_L$ \cite{NachtergaeleStarr}.
Then the higher spin analogue of the SMA calculation holds, and was reported in an earlier preprint version
of this article, available on the arXiv.
To be sure, we conjecture that the minimum energy in $S=0$ is lower than the minimum energy in the $S=1$
sectors, for the even-length spin rings, for every choice of spin $j$.
But what would be interesting is to explore if new violations occur, or if the $S=0$/$S=1$ violation is the only one for 
all $j>1/2$ on the spin ring $T_L$.

A question of some interest is to verify Dhar and Shastry's claim that the Bloch wall states are the lowest energy
eigenstates in each $S\neq 0$ sector for all spins $j$. 
That is what they claimed in \cite{DharShastry}.
I.e., it is interesting to try to find a proof that bypasses the Bethe ansatz, {\em per se}.

One may note that the commutator method gives one approach to trying to prove that the Bloch walls
are good approximate eigenstates.
Note that, writing $H_L^{\mathrm{FM}}$ as an abbreviation for $H^{\mathrm{FM}}(T_L)$,
\begin{equation}
	\left[\widehat{S}_k^+,H_{L}^{\mathrm{FM}}\right]\, 
	=\, \sum_{{\alpha}=1}^{L}   \big( e^{2\pi i k{\alpha}/L}-e^{2\pi i k({\alpha}+1)/L}\big)  
\Big(S_{\alpha}^{+} S_{[{\alpha}+1]}^{(3)}- S_{\alpha}^{(3)} S_{[{\alpha}+1]}^{+}\Big)\, , 
\end{equation}
and 
\begin{equation}
	\left[\widehat{S}_{\ell}^+,\left[\widehat{S}_k^+,H_{L}^{\mathrm{FM}}\right]\right]\, 
	=\, \sum_{{\alpha}=1}^{L}   \big( e^{2\pi i k{\alpha}/L}-e^{2\pi i k({\alpha}+1)/L}\big)  
 \big( e^{2\pi i \ell{\alpha}/L}-e^{2\pi i \ell({\alpha}+1)/L}\big) S_{\alpha}^+ S_{[\alpha+1]}^+\, .
\end{equation}
So 
$$
\left[\widehat{S}_k^+,H_{L}^{\mathrm{FM}}\right]\, (\Psi_{1/2,1/2} \otimes \cdots \otimes \Psi_{1/2,1/2})\,
=\, (1-\cos(2\pi k/L)) \widehat{S}_k^+\, (\Psi_{1/2,1/2} \otimes \cdots \otimes \Psi_{1/2,1/2})\, ,
$$
while Fourier analysis shows that 
\begin{equation*}
	\left[\widehat{S}_{\ell}^+,\left[\widehat{S}_k^+,H_{L}^{\mathrm{FM}}\right]\right]\, 
	=\, (1-e^{2\pi i k/L})(1-e^{2\pi i \ell/L}) \sum_{m=0}^{L-1} e^{-2\pi m/L} \widehat{S}^+_m \widehat{S}^+_{k+\ell-m}\, .
\end{equation*}
In the commutator framework we write:
\begin{equation*}
\begin{split}
H_L^{\mathrm{FM}} \left(\widehat{S}^+_1\right)^n\, \left(\Psi_{1/2,1/2} \otimes \cdots \otimes \Psi_{1/2,1/2}\right)\\
&\hspace{-2cm}
=\,  \sum_{r=0}^{n-1} \left(\widehat{S}^+_1\right)^r \left[H_{L}^{\mathrm{FM}},\widehat{S}_1^+\right]\, 
\left(\widehat{S}^+_1\right)^{n-1-r}\, \left(\Psi_{1/2,1/2} \otimes \cdots \otimes \Psi_{1/2,1/2}\right)\\
&\hspace{-1cm} + \left(\widehat{S}^+_1\right)^n\, H_L^{\mathrm{FM}}\,  \left(\Psi_{1/2,1/2} \otimes \cdots \otimes \Psi_{1/2,1/2}\right)\, ,
\end{split}
\end{equation*}
while also
\begin{equation*}
\begin{split}
\left(\widehat{S}^+_1\right)^r \left[H_{L}^{\mathrm{FM}},\widehat{S}_1^+\right]\, 
\left(\widehat{S}^+_1\right)^{n-1-r}\, \left(\Psi_{1/2,1/2} \otimes \cdots \otimes \Psi_{1/2,1/2}\right)\\
&\hspace{-7cm}
=\,  \sum_{s=0}^{n-r-2} \left(\widehat{S}^+_1\right)^{r+s} \left[ \left[H_{L}^{\mathrm{FM}},\widehat{S}_1^+\right],
\widehat{S}_1^+\right]
\left(\widehat{S}^+_1\right)^{n-r-2-s}\, \left(\Psi_{1/2,1/2} \otimes \cdots \otimes \Psi_{1/2,1/2}\right)\\
&\hspace{-6cm} + \left(\widehat{S}^+_1\right)^{n-1}\, \left[H_{L}^{\mathrm{FM}},\widehat{S}_1^+\right]\, \left(\Psi_{1/2,1/2} \otimes \cdots \otimes \Psi_{1/2,1/2}\right)\, .
\end{split}
\end{equation*}
But, of course, all $\widehat{S}^+_k$'s commute with each other.

We hope to carry out this exercise, elsewhere.
It may be possible that, if you can diagonalize the system in the 0, $1$ and $2$-overturned spin sectors,
then you should be able to apply this reasoning to find that the Bloch walls are approximate
eigenvectors.
This seems very similar to the coordinate Bethe ansatz approach, but the hope is that it also applies to higher spins,
since Dhar and Shastry claimed the Bloch walls still apply.

For the numerics, we are exploring two approaches. The first, obvious open direction is in using DMRG.
We have opted for exact diagonalization so far because it involved the quickest coding approach (in terms of coding time
not CPU time) and allowed us to look for new patterns.
But now that we have precise conjectures, the dramatic improvement of speed warrants using DMRG.
The second approach is paralellization.
Using a program such as RUST or C++, instead of Matlab as we have been using, will also allow a speed-up, especially
if we parallelize the program.
Alternatively, with another researcher, Sardar Jaman, we have explored quantum Monte Carlo.
That work is ongoing.

We are going to conclude with some appendices with extraneous discussions.
In Appendix \ref{sec:MP} we will summarize some additional papers coming from mathematics
and mathematical physics that impact the result.
In Appendix \ref{sec:code} we provide the basic code that we used for the exact diagonalization.
Finally in Appendix \ref{app:TL} we give two quick diagrams to suggest the proof technique for Lemma \ref{lem:first}.

\section*{Acknowledgments}

The authors gratefully acknowledge the resources provided by the University of Alabama at Birmingham IT-Research Computing group for high performance computing (HPC) support and CPU time on the Cheaha compute cluster.
S.S.~is grateful to Sardar Jaman for useful discussions.
This work was supported by a grant from the Simons Foundation.
Support provided by the National Aeronautics and Space Administration (NASA), Alabama Space Grant Consortium, Research Experiences for Undergraduates (REU) at UAB

\appendix

\section{Brief review of some important mathematical physics results}
\label{sec:MP}

Originally, David Aldous conjectured that the spectral gap of the interchange process is the same as the spectral gap of the symmetric exclusion process and also the same as the spectral gap of the random walk.
(See, for instance, the early paper of Handjani and Jungreis \cite{HandjaniJungreis}.)
That implies FOEL-$1$, and  Caputo, Liggett and Richthammer proved the full Aldous  \cite{CLR}.
Dieker also had important work related to that problem \cite{Dieker}.
Morris \cite{Morris} and also Conomos and Starr \cite{ConomosStarr} had previously proved asymptotic FOEL-$1$ on boxes.

After the proof by Caputo, Liggett and Richthammer, it was argued that asymptotic FOEL holds for boxes in \cite{NachtergaeleSpitzerStarrPre}
by Nachtergaele, Spitzer and one of the authors.
This means that for sufficiently large boxes $\{1,\dots,L\}^d$, FOEL-$n$ holds, if $L\geq L_0(n)$,
for each $n$ (where the minimum $L$ value $L_0(n)$ is a function of $n$).
If one considers the interchange process, which may be considered to be a $\mathrm{SU}(n)$ model when $L=n$,
then Alon and Kozma \cite{AlonKozma} have proved many extra interesting inequalities in the ``Aldous ordering.''
Since this model projects onto the Heisenberg model and other $\mathrm{SU}(k)$ models, for $k\leq n$, their results have implications
for all lower projections of the interchange process.

One motivation for all of this is the fact that in mathematical physics, the phase transition
for the quantum Heisenberg ferromagnet has not yet been proved, even though for the quantum Heisenberg
antiferromagnet it was proved by Dyson, Lieb and Simon \cite{DLS}.
Their method for the antiferromagnet was reflection positivity.
But that property does not hold for the quantum Heisenberg ferromagnet, as proved by Speer \cite{Speer}.
Later, Correggi, Giuliani and Seiringer did prove  the type of inequalities that one would expect from
linear spin wave analysis \cite{CGS}.
But, the phase transition for the quantum Heisenberg ferromagnet has still not been proved.
If one changes the model to add any small amount of Ising-like anisotropy it was proved by Tom Kennedy \cite{Kennedy},
but that model has $\mathrm{U}(1)$ symmetry instead of $\mathrm{SU}(2)$.

\section{The Matlab code Hulthen.m}
\label{sec:code}

The following Matlab script implements the Hulth\'en bracket basis for calculating $E^{\mathrm{FM}}_{\min}(S;1/2,L)$ for each $S$.
We have written it for $L=20$ which is the largest length spin ring we could attain with our computing resources.

\scriptsize

\begin{verbatim}
E11 = sparse([1 0; 0 0]);
E12 = sparse([0 1; 0 0]);
E21 = sparse([0 0; 1 0]);
E22 = sparse([0 0; 0 1]);
h = kron(E11,E22)+kron(E22,E11)-kron(E12,E21)-kron(E21,E12);

L=20;
fileID = fopen('L20sparse.txt','a');

dim=2^L;

Adj = diag(ones(L-1,1),1);
Adj(1,L)=1;

H = sparse(dim,dim);
for x=1:(L-1)
    for y=(x+1):L
        if Adj(x,y)
            H=H+kron(speye(2^(x-1)),kron(E11,kron(speye(2^(y-x-1)),kron(E22,speye(2^(L-y))))));
            H=H+kron(speye(2^(x-1)),kron(E22,kron(speye(2^(y-x-1)),kron(E11,speye(2^(L-y))))));
            H=H-kron(speye(2^(x-1)),kron(E12,kron(speye(2^(y-x-1)),kron(E21,speye(2^(L-y))))));
            H=H-kron(speye(2^(x-1)),kron(E21,kron(speye(2^(y-x-1)),kron(E12,speye(2^(L-y))))));
        end
    end
end

% % Cas=zeros(2^L);
% Cas = sparse(dim,dim);
% for x=1:(L-1)
%     for y=(x+1):L
%         Cas=Cas+kron(speye(2^(x-1)),kron(E11,kron(speye(2^(y-x-1)),kron(E22,speye(2^(L-y))))));
%         Cas=Cas+kron(speye(2^(x-1)),kron(E22,kron(speye(2^(y-x-1)),kron(E11,speye(2^(L-y))))));
%         Cas=Cas-kron(speye(2^(x-1)),kron(E12,kron(speye(2^(y-x-1)),kron(E21,speye(2^(L-y))))));
%         Cas=Cas-kron(speye(2^(x-1)),kron(E21,kron(speye(2^(y-x-1)),kron(E12,speye(2^(L-y))))));
%     end
% end

DnSpinMat=[];
idxLst=[];
NumBracketLst=[];
for idx = 1:dim,
    base2rep = dec2base(idx-1,2);
    LengthBase2 = length(base2rep);
    DnSpinLst=[];
    for kctr=1:LengthBase2,
        if base2rep(kctr)=='1',
            DnSpinLst = [DnSpinLst,L-LengthBase2+kctr];
        end
    end
    if length(DnSpinLst)==L/2,
        idxLst=[idxLst,idx];
        DnSpinMat=[DnSpinMat;DnSpinLst];   
    end
end
%DnSpinMat
%idxLst
vectorLst=[];
for kctr=1:length(DnSpinMat)
    bracketLstLeft=[];
    bracketLstRight=[];
    for j=1:(L/2),
        if DnSpinMat(kctr,j)>j+length(bracketLstLeft),
            bracketLstRight=[bracketLstRight,DnSpinMat(kctr,j)];
            UpSpinCompatibleLst=setdiff(setdiff(1:DnSpinMat(kctr,j),DnSpinMat(kctr,1:j)),bracketLstLeft);
            bracketLstLeft=[bracketLstLeft,max(UpSpinCompatibleLst)];
        end
    end
%    [bracketLstLeft;bracketLstRight]
    NumBracketLst=[NumBracketLst,length(bracketLstLeft)];
    v = sparse(idxLst(kctr),1,1,dim,1);
    for j=1:length(bracketLstLeft),
        a=bracketLstLeft(j);
        b=bracketLstRight(j);
        v = v-kron(speye(2^(a-1)),kron(E21,kron(speye(2^(b-a-1)),kron(E12,speye(2^(L-b))))))*v;
    end
    ExcessDnSpin = setdiff(DnSpinMat(kctr,:),bracketLstRight);
    for j=1:length(ExcessDnSpin)
        a=ExcessDnSpin(j);
        v = kron(speye(2^(a-1)),kron(E12,speye(2^(L-a))))*v;
    end
    vectorLst = [vectorLst,v];
end

for numBrktCtr = 0:(L/2)
%    numBrktCtr
    fprintf(fileID,'%d\r\n\r\n',numBrktCtr);
    Indices = find(NumBracketLst==numBrktCtr);
    vMat = vectorLst(:,Indices);
    Amat = vMat'*H*vMat;
    Bmat = vMat'*vMat;
    E0 = eigs(Amat,Bmat,1,'smallestreal','Tolerance',1e-4);
    fprintf(fileID,'%f\r\n\r\n',E0);
end
fclose(fileID);
\end{verbatim}

\normalsize
To obtain the full eigenvalue list (for sufficiently small values of $L$ such as $L=14$, where this is feasible)
change the final for-loop to this:

\scriptsize

\begin{verbatim}

for numBrktCtr = 0:(L/2)
    numBrktCtr
    Indices = find(NumBracketLst==numBrktCtr);
    vMat = vectorLst(:,Indices);
    Amat = vMat'*H*vMat;
    Bmat = vMat'*vMat;
    [V,D] = eig(full(Amat),full(Bmat),'chol');
        fprintf(fileID,'%f\r\n\r\n',diag(D));
end
fclose(fileID);
end
\end{verbatim}

\normalsize

\section{The Temperley-Lieb, Hulth\'en bracket basis}

\label{app:TL}

It uses the diagrammatic representation of Hulth\'en and of Temperley and Lieb
\cite{Hulthen,TemperleyLieb}.
For the necessary background on the diagrammatic representation, we recommend \cite{CarterFlathSaito,FrenkelKhovanov,KauffmanLins}.
The diagrammatic representation takes much time to introduce.
Therefore, we eschew it. But fundamentally the proof relies on the Perron-Frobenius theorem, similar
in spirit to the Lieb-Mattis theorem on ordering of energy levels \cite{LiebMattisOEL}.

To say a bit more, the case of proving a ferromagnetic version of the Lieb-Mattis theorem was considered
before for the case of an open chain without periodic boundary conditions.
That was in \cite{NSSfoel,NachtergaeleStarr}.
The way those theorems proceeded was in using the diagrammatic Temperley-Lieb
representation. In the Hulth\'en bracket basis, the negative of the Hamiltonian
$-\sum_{\alpha=1}^{L-1} h_{\alpha,\alpha+1}^{\mathrm{FM}}$
has all positive off-diagonal matrix elements.
Then one can apply the Perron-Frobenius theorem much like in the Lieb-Mattis theorem.

In those papers, one could not consider periodic boundary conditions because making a Temperley-Lieb generator $-h^{\mathrm{FM}}_{L,1}$,
from site $L$ to $1$, would result in crossings.
In the diagrammatic representation, crossings are resolved in  way that introduces unwanted signs which would ruin the good-signs condition necessary
for the Perron-Frobenius theorem.
In the display below, we give an example of a vector, represented diagrammatically, which results in crossings
when we apply $-h^{\mathrm{FM}}_{L,1}$:
\begin{equation*}
\begin{split}
&\begin{tikzpicture}[xscale=0.75,yscale=0.75]
	\foreach \x in {2,3,...,9}{
		\draw[very thick] (\x,0) -- (\x,-1.5);}
	\fill[black!25!white] (0.5,0) -- (10.5,0) arc (0:180:5cm and 1.5cm);
	\foreach \x in {1,3,...,7}{
		\draw[very thick] (\x,0) arc (180:0:0.5cm);}
	\draw[very thick] (9,0) -- (9,2.675) node[above] {+};
	\draw[very thick] (10,0) -- (10,2.675) node[above] {+};
	\draw[very thick] (1,0) arc (360:180:0.5cm) arc (180:0:5.5cm and 2cm) arc (360:180:0.5cm);
	\draw[very thick] (1,-1.5) .. controls (1,-0.5) and (-1,-1) .. (-1,0.5) arc (180:0:6.5cm and 2cm) .. controls (12,-1) and (10,-0.5) .. (10,-1.5);
	\foreach \x in {1,...,10} {
		\fill[white] (\x,0) circle (6pt);
		\fill (\x,0) circle (3pt);
		\fill[white] (\x,-1.5) circle (6pt);
		\fill (\x,-1.5) circle (3pt);}
	\draw[thin] (9,1.5) circle (0.25cm);
	\draw[thin] (10,1.125) circle (0.25cm);
	\draw[thin] (9,2.15) circle (0.25cm);
	\draw[thin] (10,1.9) circle (0.25cm);
\end{tikzpicture}
\end{split}
\end{equation*}

But if one restricts to singlet state vectors as the basis elements for the matrix entries
then this problem is bypassed.
In a spin singlet all the spin sites on the top row are paired together in a non-crossing
perfect matching. In particular there are no free strands extending to the top edge of the diagram
(or spanning from the top row to the bottom row in some conventions).
Then one may join left to right by a Temperley-Lieb generator without any crossings.
But this only works for the spin-singlet sector.

An example of a calculation is shown to demonstrate this point:
it shows the result of appling $-h^{\mathrm{FM}}_{L,1}$ to a particular Hulthen
bracket basis element, using the diagrammatic representation
\begin{equation*}
\begin{split}
&\begin{tikzpicture}[xscale=0.75,yscale=0.75]
	\foreach \x in {2,3,...,9}{
		\draw[very thick] (\x,0) -- (\x,-1.5);}
	\fill[black!25!white] (0.5,0) -- (10.5,0) arc (0:180:5cm and 1.5cm);
	\foreach \x in {1,3,...,9}{
		\draw[very thick] (\x,0) arc (180:0:0.5cm);}
	\draw[very thick] (1,0) arc (360:180:0.5cm) arc (180:0:5.5cm and 2cm) arc (360:180:0.5cm);
	\draw[very thick] (1,-1.5) .. controls (1,-0.5) and (-1,-1) .. (-1,0.5) arc (180:0:6.5cm and 2cm) .. controls (12,-1) and (10,-0.5) .. (10,-1.5);
	\foreach \x in {1,...,10} {
		\fill[white] (\x,0) circle (6pt);
		\fill (\x,0) circle (3pt);
		\fill[white] (\x,-1.5) circle (6pt);
		\fill (\x,-1.5) circle (3pt);}
\end{tikzpicture}\\[10pt]
&\hspace{3cm} =\,
\raisebox{0.5cm}{
\begin{minipage}{8cm}
\begin{tikzpicture}[xscale=0.75,yscale=0.75]
%	\foreach \x in {2,3,...,9}{
%		\draw[very thick] (\x,0) -- (\x,-1.5);}
%	\fill[black!35!white] (0.5,0) -- (10.5,0) arc (0:180:5cm and 1.5cm);
	\foreach \x in {3,5,7}{
		\draw[very thick] (\x,0) arc (180:0:0.5cm);}
	\draw[very thick] (2,0) .. controls (2,2) and (9,2) .. (9,0);
	\draw[very thick] (1,0) .. controls (1,3) and (10,3) .. (10,0);
	\foreach \x in {1,...,10} {
		\fill[white] (\x,0) circle (6pt);
		\fill (\x,0) circle (3pt);}
\end{tikzpicture}
\end{minipage}}
\end{split}
\end{equation*}
No matter how one pairs up the spins in the grey region into any non-crossing perfect matching,
the resulting diagram will either be: 
another different Hulthen bracket basis vector, with coefficient 1;
or else a ``bubble'' resulting in a negative coefficient $-2$ but with the same
diagrammatric basis element, yielding a contribution to a diagonal matrix element
(which is allowed to be negative in the Perron-Frobenius theorem).
So all off-diagonal matrix entries are positive in the Hulthen bracket basis for the negative
of the Hamiltonian. So the Perron-Frobenius theorem applies.

\baselineskip=12pt
\bibliographystyle{plain}

\begin{thebibliography}{10}
%
%\bibitem{AffleckLieb}
%Ian Affleck and Elliott H.~Lieb.
%\newblock A Proof of Part of Haldane's Conjecture on Spin Chains.
%\newblock {\em Lett.~Math.~Phys.} {\bf 12}, 57--69 (1986).
%
%\bibitem{AizenmanNachtergaele}
%Michael Aizenman and Bruno Nachtergaele.
%\newblock Geometric Aspects of Quantum Spin Systems.
%\newblock {\em Commun.~Math.~Phys.} {\bf 164}, no.~1, 17--63 (1994).

\bibitem{NSSfoel} % ref 1
Bruno Nachtergaele, Wolfgang Spitzer and Shannon Starr.
\newblock Ferromagnetic Ordering of Energy Levels.
\newblock {\em J.~Statist.~Phys.} {\bf 116}, 719--738 (2004).

\bibitem{LiebMattisOEL} % ref 2
Elliott Lieb and Daniel Mattis.
\newblock Ordering of Energy Levels of Interacting Spin Systems.
\newblock {\em J.~Mathem.~Phys.} {\bf 3}, no.~4, 749--751 (1962).

\bibitem{SpitzerStarrTran} % ref 3
Wolfgang Spitzer, Shannon Starr and Lam Tran.
\newblock Counterexamples to Ferromagnetic Ordering of Energy Levels.
\newblock {\em J.~Math.~Phys.} {\bf 53}, no.~4, 043302 (2012).

\bibitem{TemperleyLieb} % ref 4
H.~N.~V.~Temperley and Elliott H.~Lieb.
\newblock Relation Between the `Percolation' and `Colouring' Problem,
and Other Graph Theoretical Problems Associated with Regular Planar Lattices:
Some Exact Results for the `Percolation' Problem.
\newblock {\em Proc.~Royal Soc.~London A} {\bf 322}, 251--280 (1971). 

\bibitem{DharShastry} % ref 5
Abishek Dhar and B.~Sriram Shastry.
\newblock Bloch Walls and Macroscopic String States in Bethe's Solution of the Heisenberg Ferromagnetic Linear Chain.
\newblock {\em Phys.~Rev.~Lett.} {\bf 85}, 2813 (2000).

\bibitem{Auerbach} % ref 6
Assa Auerbach.
\newblock {\em Interacting Electrons and Quantum Magnetism.}
\newblock Springer-Verlag, New York 1994.

\bibitem{NachtergaeleStarrDroplet} % ref 7 
Bruno Nachtergaele and Shannon Starr.
\newblock Droplet States in the XXZ Heisenberg Chain.
\newblock {\em Commun.~Math.~Phys.} {\bf 218}, 569--607 (2001).

\bibitem{NSSdroplet} % ref 8
Bruno Nachtergaele, Wolfgang Spitzer and Shannon Starr.
\newblock Droplet Excitations for the Spin 1/2 XXZ Chain with Kink Boundary Conditions.
\newblock {\em Ann.~Henri~Poincar\'e} {\bf 8} 165--201 (2007).

\bibitem{Sutherland} % ref 9 
Bill Sutherland.
\newblock{Low-Lying Eigenstates of the One-Dimensional Heisenberg Ferromagnet for any Magnetization and Momentum.}
\newblock {\em Phys.~Rev.~Lett.} {\bf 75}, no.~5, 816--819 (1995).

\bibitem{Bethe} % ref 10
Hans Bethe.
\newblock Zur Theorie der Metalle.
I. Eigenwerte und Eigenfunktionen der linearen Atomkette.
\newblock {\em Zeit.~Physik} {\bf 71} 205--226 (1931).

\bibitem{BBG} % ref 11
Till Bargheer, Niklas Beisert and Nikolay Gromov.
\newblock Quantum stability for the Heisenberg ferromagnet.
\newblock {\em New J.~Phys.} {\bf 10}, 103023 (76 pages) (2008).

\bibitem{CGS} % ref 12
Michele Correggi, Alessandro Giuliani and Robert Seiringer.
\newblock Validity of the Spin-Wave Approximation for the Free Energy of the Heisenberg Ferromagnet.
\newblock {\em Commun.~Math.~Phys.} {\bf 339}, 279--307 (2015).

\bibitem{NachtergaeleStarr} % ref 13
Bruno Nachtergaele and Shannon Starr.
\newblock Ferromagnetic Lieb-Mattis Theorem.
\newblock {\em Phys.~Rev.~Lett.} {\bf 94}, 057206 (2005).

\bibitem{HandjaniJungreis} % ref 14
Shirin Handjani and Douglas Jungreis.
\newblock Rate of Convergence for Shuffling Cards by Transpositions.
\newblock {\em J.~Theoret.~Probab.} {\bf 9}, 983--993 (1996).

\bibitem{CLR} 
Thomas Liggett, Pietro Caputo and Thomas Richthammer.
\newblock Proof of Aldous' Spectral Gap Conjecture.
\newblock {\em J.~Amer.~Math.~Soc.} {\bf 23}, no.~3, 831--851 (2010).

\bibitem{Dieker} % ref 15
 A.~B.~Dieker.
\newblock Interlacings for Random Walks on Weighted Graphs and the Interchange Process. 
\newblock {\em SIAM J.~Discrete Math} {\bf 24}, no.~1, 191--206 (2010).

\bibitem{Morris} % ref 16
Ben Morris.
\newblock Spectral Gap for the Interchange Process in a Box.
\newblock {\em Electron.~Commun.~Probab.} {\bf 13}, 311--318 (2008).

\bibitem{ConomosStarr} % ref 17
Matt Conomos and Shannon Starr.
Asymptotics of the Spectral Gap for the Interchange Process on Large Hypercubes.
\newblock {\em J.~Statist.~Mech.} {\bf 2011}, P10018 (2011).

\bibitem{NachtergaeleSpitzerStarrPre} % ref 18
Bruno Nachtergaele, Wolfgang Spitzer and Shannon Starr.
\newblock Asymptotic Ferromagnetic Ordering of Energy Levels for the Heisenberg Model on Large Boxes.
\newblock {\em Preprint}, \url{https://arxiv.org/abs/1509.00907}, 2015.

\bibitem{AlonKozma} % ref 19
Gil Alon and Gady Kozma.
\newblock Ordering the Representations of ${{S}_{n}}$  Using the Interchange Process.
\newblock {\em Canad.~Math.~Bull.} {\bf 56}, no.~1, 13--30 (2013).

\bibitem{DLS} % ref 20
Freeman J.~Dyson, Elliott H.~Lieb and Barry Simon.
\newblock Phase Transitions in Quantum Spin Systems with Isotropic and Nonisotropic Interactions.
\newblock {\em J.~Statist.~Phys.} {\bf 18}, 335--383 (1978).

\bibitem{DLS} % ref 21
Freeman J.~Dyson, Elliott H.~Lieb and Barry Simon.
\newblock Phase Transitions in Quantum Spin Systems with Isotropic and Nonisotropic Interactions.
\newblock {\em J.~Statist.~Phys.} {\bf 18}, 335--383 (1978).

\bibitem{Speer} % ref 22
Eugene~R.~Speer.
\newblock Failure of Reflection Positivity in the Quantum Heisenberg Ferromagnet.
\newblock {\em Lett.~Math.~Phys.} {\bf 10}, 41--47 (1985).

\bibitem{Kennedy} % ref 23
Tom Kennedy.
\newblock Long Range Order in the Anisotropic Quantum Ferromagnetic Heisenberg Model.
\newblock {\em Comm.~Math.~Phys.} {\bf 100}, no.~3, 447--462 (1985).

\bibitem{Hulthen} % ref 24
Lamek Hulth\'en.
\newblock  \"Uber das Austauschproblem eines Kristalles (thesis).
\newblock {\em Arkiv f\"or Mat.~Astrom.~och.~Fisik}, {\bf 26A}, no.~11, pp.~1--106 (1938). 

\bibitem{CarterFlathSaito} % ref 25
J.~Scott Carter, Daniel E.~Flath and Masahico Saito.
\newblock {\em The Classical and Quantum 6j-Symbols.}
\newblock Princeton University Press, Princeton, NJ 1991.

\bibitem{FrenkelKhovanov} % ref 26
Igor B.~Frenkel and Mikhail B.~Khovanov.
\newblock Canonical Bases in Tensor Products and Graphical Calculus for $U_q(\mathfrak{sl}_2)$.
\newblock {\em Duke Math.~J.} {\bf 87}, no.~3, 409--480 (1997).

\bibitem{KauffmanLins} % ref 27
Louis H.~Kauffman and S\'ostenes L.~Lins.
\newblock {Temperley-Lieb Recoupling Theory and Invariants of 3-Manifolds.}
\newblock Princeton University Press, Princeton, NJ 1994.

%%%% UNUSED REFERENCES %%%%
%
%\bibitem{KlugeRubey}
%\newblock Stefan Kluge and Martin Rubey.
%\newblock Cyclic Sieving for torsion pairs in the cluster category of Dynkin type $A_n$.
%\newblock {\em Preprint} 2011, \url{https://arxiv.org/abs/1101.1020}.
%
%\bibitem{BKI}
%V.~E.~Korepin, N.~M.~Bogoliubov and A.~G.~Izergin.
%\newblock {\em Quantum Inverse Scattering Method and Correlation Functions.}
%\newblock Cambridge University Press, UK 1997.
%
%\bibitem{GJL} 
%Olivier Golinelli, Thiery Jolicoeur and Robert Lacaze.
%\newblock Heisenberg Antiferromagnetic Chain of Spin $S=1$.
%\newblock {\em Intern.~J.~Modern Phys.~C} {\bf 5}, 259--261 (1994).
%
%\bibitem{Charalambides}
%Charalambos A.~Charalambides.
%\newblock {\em Enumerative Combinatorics}.
%\newblock CRC Press, Taylor \& Francis group, Boca Raton, FL, 2002.

%
%\bibitem{CrawfordNgStarr}
%Nicholas Crawford, Stephen Ng and Shannon Starr.
%\newblock Emptiness Formation Probability.
%\newblock {\em Commun.~Math.~Phys.} {\bf 345}, no.~3, 881922 (2016).


%\bibitem{Dyson1}
%Freeman J.~Dyson.
%\newblock General Theory of Spin-Wave Interactions.
%\newblock {\em Phys.~Rev.} {\bf 102}, no.~5, 1217--1230 (1956).
%


%\bibitem{Haldane}
%F.~D.~M.~Haldane.
%\newblock Continuum Dynamics of the 1-D Heisenberg Antiferromagnet: Identification with the O(3) Nonlinear Sigma Model.
%\newblock {\em Phys.~Lett.~A} {\bf 93}, no.~9, 464--468 (1983).

%\bibitem{Hastings}
%Matthew B.~Hastings.
%\newblock Lieb-Schultz-Mattis in Higher Dimensions.
%\newblock {\em Phys.~Rev.~B} {\bf 69}, 104431 (2004).
%
%\bibitem{HastingsKoma}
%Matthew B.~Hastings and Tohru Koma.
%\newblock Spectral Gap and Exponential Decay of Correlations.
%\newblock {\em Commun.~Math.~Phys.} {\bf  265}, 781--804 (2006).

%
%\bibitem{LiebSchultzMattis}
%Elliott Lieb, Theodore Schultz and Daniel Mattis.
%\newblock Two Soluble Models of an Antiferromagnetic Chain.
%\newblock {\em Ann.~Physics} {\bf 16}, 407--466 (1961).
%
%\bibitem{Liggett}
%Thomas M.~Liggett.
%\newblock {\em Interacting Particle Systems: Reprint of the 1985 Edition with a New Postface.}
%\newblock Springer Verlag, Berlin, Germany 2005.


%
%\bibitem{Toth}
%Balint Toth.
%\newblock Improved Lower Bounds on the Thermodynamic Pressure of the Spin 1/2 Heisenberg Ferromagnet.
%\newblock {\em Lett.~Math.~Phys.} {\bf 28}, 75--84 (1993).

\end{thebibliography}

\end{document}